\documentclass{llncs}
\pdfoutput=1

\usepackage{amsmath, amssymb}
\usepackage{graphicx} 
\usepackage[utf8]{inputenc}
\usepackage[font=small,labelfont=bf]{caption}
\usepackage[font=small,labelfont=small]{subcaption}
\usepackage[bookmarksopen,bookmarksdepth=3]{hyperref}
\usepackage{ifluatex}
\usepackage[shortlabels]{enumitem}
\ifluatex
  \usepackage{fontspec}
  \defaultfontfeatures{Ligatures=TeX}
  \usepackage[splitting=weightedMedian]{luatodonotes}
\else
  \usepackage{todonotes}
\fi
\usepackage{tikz}
\usetikzlibrary{calc}
\usepackage{wrapfig}
\usepackage{ifthen}
\usepackage{timetravel}
\usepackage{xspace}
\setCounters{thm,prop,lem,cor,obs}

\usepackage{arxiv-appendix}
\arxivid{1607.00278v2}
\arxivcite{ARXIV}
\arxivauxname{obstacle_number-arxiv}

\iflncs
\else
\fi

\title{Obstructing Visibilities with One Obstacle%
\iflncs%
  \thanks{\arxivrefthanks}
\else%
  \thanks{Appears in the Proceedings of the 24th International Symposium on Graph Drawing and Network Visualization (GD 2016).}
\fi%
}

\author{Steven Chaplick\inst{1} \and Fabian Lipp\inst{1} \and Ji-won
  Park\inst{2}\thanks{J.P.\ acknowledges support by the NRF grant
    2011-0030044 (SRC-GAIA) funded by the Korean government.}  \and
  Alexander Wolff\inst{1}}

\institute{%
Lehrstuhl f\"ur Informatik~I, Universit\"at W\"urzburg, Germany
\and
KAIST, Korea}

\newtheorem{thm}{Theorem}

\newtheorem{lem}{Lemma}
\newtheorem{cor}{Corollary}
\newtheorem{obs}{Observation}

\newcommand{\WLOG}{Without loss of generality, }
\newcommand{\wLOG}{without loss of generality, }

\DeclareMathOperator{\Obs}{obs}
\DeclareMathOperator{\Circ}{circ}
\DeclareMathOperator{\CH}{CH}
\DeclareMathOperator{\Obsin}{\Obs_\mathrm{in}}
\DeclareMathOperator{\Obsout}{\Obs_\mathrm{out}}
\newcommand{\Gin}{\ensuremath{G_\mathrm{in}}\xspace}

\let\doendproof\endproof
\renewcommand\endproof{~\hfill$\qed$\doendproof}

\graphicspath{{graphs/}{figures/}}

\begin{document}

\maketitle

\begin{abstract}
  Obstacle representations of graphs have been investigated quite
  intensely over the last few years.  We focus on graphs that can
  be represented by a single obstacle.  
  Given a (topologically open) non-self-intersecting polygon~$C$ and a finite set~$P$
  of points in general position in the complement of~$C$, the
  \emph{visibility graph}~$G_C(P)$ has a vertex for each 
  point in~$P$ and an edge~$pq$ for any two
  points~$p$ and~$q$ in~$P$ that can \emph{see} each other, that is,
  $\overline{pq} \cap C=\emptyset$.  We draw $G_C(P)$ straight-line
  and call this a \emph{visibility drawing}.
  Given a graph~$G$, we want to compute an obstacle representation
  of~$G$, that is, an obstacle~$C$ and a set of points~$P$ such that
  $G=G_C(P)$.  The complexity of this problem is open, even when
  the points are exactly the vertices of a simple polygon
  and the obstacle is the complement of the polygon---the
  \emph{simple-polygon visibility graph problem}.

  ~~~~There are two types of obstacles; \emph{outside} obstacles lie in the unbounded
component of the visibility drawing, whereas \emph{inside} obstacles
lie in the complement of the unbounded component.
  We show
  that the class of graphs with an inside-obstacle representation is
  incomparable with the class of graphs that have an outside-obstacle
  representation. We further show that any graph with at most seven
  vertices has an outside-obstacle
  representation, which does not hold for a specific graph with eight
  vertices.
  Finally, we show NP-hardness of the \emph{outside-obstacle graph sandwich
    problem}: given graphs $G$ and~$H$ on the same vertex set, is
  there a graph~$K$ such that $G \subseteq K \subseteq H$ and $K$ has
  an outside-obstacle representation.
  Our proof also shows that the
  \emph{simple-polygon visibility graph sandwich problem},
  the \emph{inside-obstacle graph sandwich problem}, and
  the \emph{single-obstacle graph sandwich problem} are all NP-hard.
\end{abstract}

\section{Introduction}

Recognizing graphs that have a certain type of geometric
representation is a well-established field of research dealing with,
for example, interval graphs, unit disk graphs, coin graphs (which are
exactly the planar graphs), and visibility graphs.  In this paper, we
are interested in visibilities of points in the presence of a single
obstacle.  Given a (topologically open) non-self-intersecting
polygon~$C$ 
and a finite set~$P$ of points in general position in the complement of~$C$,
the \emph{visibility graph}~$G_C(P)$ has a vertex for each point
in~$P$ and an edge~$pq$ for any two points~$p$ and~$q$ in~$P$ that can
\emph{see} each other, that is, $\overline{pq} \cap C=\emptyset$.
Given a graph~$G$, we want to compute a (single-) obstacle
representation of~$G$, that is, an obstacle~$C$ and a set of
points~$P$ such that $G=G_C(P)$ (if such a representation exists).
The complexity of this reconstruction problem is open, even for the case that
the points are exactly the vertices of a simple polygon
and the (outside) obstacle is the complement of the polygon.  
This special case is called the \emph{simple-polygon visibility graph (reconstruction) problem}.

The \emph{visibility drawing} is a straight-line drawing of the
visibility graph.  The visibility drawing allows us to differentiate
two types of obstacles: \emph{outside} obstacles lie in the unbounded
component of the visibility drawing, whereas \emph{inside} obstacles
lie in the complement of the unbounded component. 

If we drop the restriction to single obstacles, our problem can 
be seen as an optimization problem. 
For a graph $G$, let $\Obs(G)$ be the smallest number of obstacles
that suffices to represent~$G$ as a visibility graph.  Analogously,
let $\Obsout(G)$ be the number of obstacles needed to represent~$G$ in
the presence of an outside obstacle, and let $\Obsin(G)$ be the number
of obstacles needed to represent~$G$ in the absence of outside
obstacles.
Specifically, we say that $G$ has an \emph{outside-obstacle
  representation} if $G$ can be represented by a single outside
obstacle (e.g. Fig.~\ref{fig:4+5cycle}), and $G$ has an
\emph{inside-obstacle representation} if $G$ can be represented by a
single inside obstacle (e.g. Fig.~\ref{fig:K}).

\noindent\textit{\textbf{Previous work.}}
Not only have Alpert et al.~\cite{akl-ong-DCG10} introduced the notion
of the obstacle number of a graph, they also characterized the 
class of graphs that can be represented by a single simple obstacle,
namely a convex polygon.
They also asked many interesting questions, for example, given an integer $o$,
is there a graph of obstacle number exactly $o$?
If the previous question is true, given an integer $o>1$, what is the smallest number of vertices of a graph with obstacle number~$o$? 
Mukkamala et al.~\cite{mps-glon-WG10} 
showed the first question is true. For the second question,
Alpert et al.~\cite{akl-ong-DCG10} found a 12-vertex graph that needs
two obstacles, namely $K^*_{5,7}$, where $K^*_{m,n}$
with $m \le n$ is the complete bipartite graph minus a matching of
size~$m$. 
They also showed that for any $m \le n$, $\Obs(K^*_{m,n}) \le~2$.
This result was improved by Pach and Sar\i\"{o}z
\cite{ps-sglon-GC11} who showed that the 10-vertex graph~$K^*_{5,5}$
also needs two obstacles.
More recently, Berman et al.~\cite{bcfghw-gong1-Manu16}
suggested some necessary conditions for a graph to have obstacle number 1 which they used to find a
\emph{planar} 10-vertex graph that cannot be represented by a single
obstacle. 

Alpert et al.~\cite{akl-ong-DCG10} conjectured that every graph of
obstacle number~1 has also outside-obstacle number~1.  Berman et
al.~\cite{bcfghw-gong1-Manu16} further conjectured that every graph of
obstacle number~$o$ has outside-obstacle number~$o$.
Alpert et al.~\cite{akl-ong-DCG10} also showed that outerplanar graphs
always have outside-obstacle representations and posed the question to
bound the inside/convex obstacle number of outerplanar/planar graphs.
Fulek et al.~\cite{fss-conog-30EGGT} partly answered this 
by showing that five convex obstacles are sufficient
for outerplanar graphs---and that sometimes four are needed.

For the asymptotic bound on the obstacle number of a graph,
it is obvious that any $n$-vertex graph has obstacle number~$O(n^2)$.
Balko et al.~\cite{bcv-dgusno-GD15} showed that the obstacle number of
an $n$-vertex graph is (at most) $O(n\log n)$.
For the lower bound, improving on previous
results \cite{akl-ong-DCG10,mps-glon-WG10,mpp-lbong-EJC12},
Dujmovi\'{c} and Morin~\cite{dm-on-EJC15} showed there are $n$-vertex
graphs whose obstacle number is $\Omega(n/(\log\log n)^2)$.

Johnson and Sar\i\"oz \cite{js-rpslg-CCCG14} investigated the special
case where the visibility graph is required to be plane.  They showed
(by reduction from \textsc{PlanarVertexCover}) that in this case
computing the obstacle number is NP-hard.  By 
reduction to \textsc{Maxdeg-3 PlanarVertexCover}, they showed that
the problem admits a polynomial-time approximation scheme and is
fixed-parameter tractable.
Koch et al.~\cite{kkr-gpoor-arXiv13} also considered the plane case,
restricted to outside obstacles.  They gave a(n efficiently checkable)
characterization of all biconnected graphs that admit a plane
outside-obstacle representation.

A few years ago, Ghosh and Goswami~\cite{gg-upvgp-arXiv12} surveyed
visibility graph problems, among them
simple-polygon visibility graph problem.
Open Problem~29 in their survey is the complexity of the recognition problem
and Open Problem~33 is the complexity of the fore-mentioned reconstruction problem.
Very recently, this question has been settled for an interesting
variant of the problem where the points are not only the vertices of
the graph but also the obstacles (which are closed in this
case): Cardinal and Hoffmann~\cite{ch-rcpvg-SoCG15} showed that
recognizing point-visibility graphs is $\exists\mathbb{R}$-complete, that
is, as hard as deciding the existence of a real solution to a system
of polynomial inequalities (and hence, at least NP-hard).

The graph sandwich problem has been introduced by Golumbic et
al.~\cite{gks-gsp-JA95} as a generalization of the recognition
problem.  They set up the abstract problem formulation and gave
efficient algorithms for some concrete graph properties---and hardness
results for others.

\noindent\textit{\textbf{Preliminaries.}}
In this paper, we consider only finite simple graphs.
Whenever we say cycles, we always mean simple cycles. Let $G$ be a graph and let $v, u$ be its vertices. The
\emph{circumference} of $G$, denoted by $\Circ(G)$, is the length
of its longest cycle. $v\sim u$ denotes that $v$ and $u$ are adjacent.
We call $v$ and $u$ \emph{twins} if $v \neq u$ and $N(v)\backslash\{u\} =
N(u)\backslash\{v\}$.
We say $v$ is \emph{exposed to the outside} if it is on the boundary of the unbounded component of the straight-line drawing of $G$ given by the point set.
All vertices are exposed to the outside in an \emph{exposed outside-obstacle representation}.
In all figures (of graphs), unless otherwise stated, edges are solid and non-edges are dashed.

\noindent\textit{\textbf{Our contribution.}}
We have the following results. (Recall that a \emph{co-bipartite}
graph is the complement of a bipartite graph.)
\begin{itemize}[topsep=0pt]
\item Every graph of circumference at most 6 has an outside-obstacle 
representation (Theorem~\ref{thm graphs circumference up to 6}).
\item Every 7-vertex graph has an outside-obstacle representation
  (Theorem~\ref{thm graphs order up to 7}). Moreover, there is an
  8-vertex co-bipartite graph that has no single-obstacle
  representation (Theorem~\ref{thm smallest graph}). 
\item There is an 11-vertex co-bipartite graph with an
  inside-obstacle representation, but no outside-obstacle
  representation (Theorem~\ref{thm:inside||outside}).  This resolves
  the above-mentioned open problems of Alpert et
  al.~\cite{akl-ong-DCG10} and Berman et al.~\cite{bcfghw-gong1-Manu16}.
\item The Outside-Obstacle Graph Sandwich Problem is NP-hard even for
  co-bipartite graphs.  The same holds for the Simple-Polygon
  Visibility Graph Sandwich Problem.  This does not solve, but sheds
  some light on a long-standing open problem: the recognition of
  visibility graphs of simple polygons.
  While little is known for the complexity of computing the obstacle number,
  the Single-Obstacle Graph Sandwich Problem is shown to be also NP-hard.
\end{itemize}

\noindent\textit{\textbf{Remarks and Open Problems.}}
The recognition of inside- and outside-obstacle graphs is currently
open.  We expect that testing either of these cases is NP-hard.
Assuming that this is true, it would be interesting to show
fixed-parameter tractability w.r.t.\ the number of vertices of the
obstacle.
We now know that $\Obsin(G)$ and $\Obsout(G)$ are usually different,
but can we bound $\Obsin(G)$ in terms of $\Obsout(G)$?  While we have
shown that the trivial lower bound $\Obsout(G)-1$ is tight, an upper
bound is only known for outerplanar graphs
\cite{akl-ong-DCG10,fss-conog-30EGGT}.

\section{Graphs with Small Circumference}
\label{sec:small_circ}

In this section we will describe how to construct an outside-obstacle 
representation for any graph whose circumference is at most 6.
To prove this result we show that for every vertex $v$ of a
biconnected graph~$G$ with circumference at most~6, 
there is an exposed outside-obstacle representation of $G$ with $v$ on the convex hull of $V(G)$.
Lemma~\ref{lem nonadjacent twin} makes it easier to describe the outside-obstacle representation.
We then apply Lemma~\ref{lem merge graphs} and Lemma~\ref{lem identifying graphs} to obtain an outside-obstacle representation of a graph.

We provide an 8-vertex graph of circumference~8 that requires at least
two obstacles in the next section, so the only gap is the
circumference-7 case.
We conjecture that every graph of circumference~7 has an
outside-obstacle representation. 
As a first step towards this conjecture, we show that every 7-vertex
graph has an outside-obstacle representation by providing
a list of point sets such that each 7-vertex graph can be represented
by an outside obstacle when the vertices of the graph are
mapped to a point set in our list.

Proofs of Lemmas~\ref{lem merge graphs},\ref{lem identifying graphs},\ref{lem nonadjacent twin} are in Appendix~\appref{appendixA} and brief ideas are sketched here.

\newcommand{\LemMergeGraphsText}{%
Let $G$ and $H$ be graphs on different vertex sets.
If $\Obsout(G)=1$ and $\Obsout(H)=1$, then $\Obsout(G \cup H)=1$.}
\wormhole{lem merge graphs}
\begin{lem}\label{lem merge graphs}
\LemMergeGraphsText
\end{lem}
\begin{proof}[sketch]
Place two graphs far enough and merge outside obstacles.
\end{proof}

\newcommand{\LemIdentifyingGraphsText}{%
  Let $G$ and~$H$ be graphs with exposed outside-obstacle representations.
  Let $u$ be a vertex
  of~$G$, and let~$v$ be a vertex of~$H$.  Assume that $v$ lies on the
  convex hull of~$V(H)$.  If $K$ is the graph obtained by
  identifying~$u$ and~$v$, then $K$ also has an exposed outside-obstacle
  representation.}
\wormhole{lem identifying graphs}
\begin{lem}\label{lem identifying graphs}
\LemIdentifyingGraphsText
\end{lem}
\begin{proof}[sketch]
Make the outside-obstacle representation of $H$ small and narrow (with respect to $v$) enough to fit in some circular sector lying inside the obstacle centered at $u$ in the outside-obstacle representation of $G$.
Then replace the circular sector with above obstacle representation of $H$.
\end{proof}

\newcommand{\LemNonadjacentTwinText}{%
  Let $H$ be a graph, $v$ be a vertex of~$H$, $A$ be the set of twins of~$v$, and $G = H \setminus A$.
  If $G$ that has an exposed outside-obstacle representation in which $v$ lies on the convex hull of $V(G)$,
  then $H$ has an exposed outside-obstacle representation in which all vertices in $A \cup \{v\}$ lie on the convex hull of $V(H)$.} 
\wormhole{lem nonadjacent twin}
\begin{lem}\label{lem nonadjacent twin}
\LemNonadjacentTwinText
\end{lem}
\begin{proof}[sketch]
Place twins close enough since their neighborhoods are same.
\end{proof}

The following observation helps to restrict the structure of
biconnected graphs of given circumference
where indices are taken modulo $k$.
\begin{obs}\label{obs cycle and path}
    Let $G$ be a graph of circumference $k$ and let $C=v_1v_2\dots v_k$ be a cycle.
    $G$ doesn't contain a $v_i-v_{i+t}$ path $P$ of length $t'$ disjoint to $v_iCv_{i+t}$
    where $0<t<k$ and $t'>t$, since it would create $(k+t'-t)$-cycle.
    In particular, if $v \notin \{v_1, \dots, v_k\}$ is adjacent to $v_i$,
    then $v$ is neither adjacent to $v_{i-1}$ nor $v_{i+1}$.
\end{obs}

\newcommand{\ThmCircumferenceUpToSixText}{%
  If the circumference of a graph $G$ is at most 6, then $G$ has an
  outside-obstacle representation.}
\wormhole{thm graphs circumference up to 6}
\begin{thm}\label{thm graphs circumference up to 6}
  \ThmCircumferenceUpToSixText
\end{thm}

\begin{wrapfigure}[10]{r}{.34\textwidth}
  \vspace{-5ex}
  \centering
  \begin{tabular}[b]{@{}c@{\quad}c@{}}
    \includegraphics{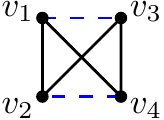} & \includegraphics{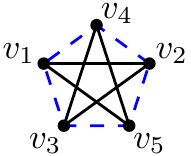}
    \\[-.2ex]
    \small (a) & \small (b) \\
  \end{tabular}
  \caption{Graphs of circumference~4 and~5 with outside-obstacle
    representations}
  \label{fig:4+5cycle}
\end{wrapfigure}  

\noindent\emph{Proof.}~
  If $G$ is disconnected, we give an outside-obstacle representation
  for each connected component and simply merge them by Lemma~\ref{lem
    merge graphs}. 
  
  When $G$ is connected, we decompose it into its biconnected components, 
  i.e., the block decomposition tree of $G$.
  Starting in its root, we include representations of the children in 
  turn using Lemma~\ref{lem identifying graphs}.
    
  Let $H$ be a biconnected component of $G$. It suffices to show that
  $H$ satisfies the condition for Lemma~\ref{lem identifying graphs}:
  For each vertex~$v$ of~$H$, $H$ has an exposed outside-obstacle
  representation such that $v$ is on the convex hull of $V(H)$. 

\medskip
\noindent\textbf{Case~1:} $\Circ(H) = 3$

     As $H$ is biconnected, $H$ is a triangle and trivially satisfies the condition.
  
\medskip
\noindent\textbf{Case~2:} $\Circ(H) = 4$
  
    Let $C=v_1v_2v_3v_4 \subset H$ be a 4-cycle.
    If $H$ contains exactly four vertices, there is an outside-obstacle
    representation; see Fig.~\ref{fig:4+5cycle}a.
    Note that we can choose the (dashed blue) diagonals $v_1v_3$ and
    $v_2v_4$ to be edges or non-edges as desired.
    Otherwise, \wLOG there is a vertex $x \in H \setminus C$ with $x \sim v_1$.
    As $H$ is biconnected, there is a path of length at least~2 from $v_1$ to
    another vertex of $C$ containing $x$.
    Observation~\ref{obs cycle and path} implies that $x \not\sim v_2$,
    $x \sim v_3$, and $x \not\sim v_4$.
    Since we have another 4-cycle $C'=v_1xv_3v_4$, the same holds for $v_2$,
    implying $v_2 \not\sim v_4$. Hence $x$ is a non-adjacent twin of $v_2$.
    It follows that any vertex in $H \setminus C$ is a non-adjacent
    twin of one of $v_1, \dots, v_4$.  Since the vertices in
    Fig.~\ref{fig:4+5cycle}a are in convex position, we can
    embed~$H$ using Lemma~\ref{lem nonadjacent twin}.
   
\medskip
\noindent\textbf{Case~3:} $\Circ(H) = 5$
  
    Let $C=v_1v_2v_3v_4v_5 \subset H$ be a 5-cycle.
    If $H$ contains exactly five vertices, see Fig.~\ref{fig:4+5cycle}b
    for its outside-obstacle representation.
    Otherwise, \wLOG there is a vertex $x \in H \setminus C$ with $x \sim v_1$.
    Observation~\ref{obs cycle and path} implies $x \not\sim v_2, v_5$.
    As $H$ is biconnected, there is either path $v_1xv_3$ or $v_1xv_4$.
    \WLOG we assume $x \sim v_3$ and thus $x \not\sim v_4$.
    Then $v_2 \not\sim v_4, v_5$ since we have another 5-cycle
    $v_1xv_3v_4v_5$ and can apply the same logic.
    Hence, $x$ is a non-adjacent twin of $v_2$.
    As in the Case~2, we see that every vertex in $H \setminus
    C$ is a non-adjacent twin of one of $v_1, v_2, \dots, v_5$ and we can
    embed $H$ using Lemma~\ref{lem nonadjacent twin}.
    
\medskip
\noindent\textbf{Case~4:} $\Circ(H) = 6$~~
   (We postpone this case to Appendix~\appref{appendixA}.) \hfill\qed

\newcommand{\ThmGraphsOrderUpToSevenText}{%
  Any graph with at most 7 vertices has an outside-obstacle
  representation.}
\wormhole{thm graphs order up to 7}
\begin{thm}\label{thm graphs order up to 7}
\ThmGraphsOrderUpToSevenText
\end{thm}
\begin{proof}[sketch]
  By Theorem~\ref{thm graphs circumference up to 6}, it suffices to
  provide an outside-obstacle representation of each 7-vertex graph containing $C_7$. 
  In Appendix~\appref{appendixA}, we classify such graphs into 15 groups and 
  give an outside-obstacle representation of each. 
\end{proof}

\section{Co-Bipartite Graphs}
\label{sec:co-bipartite}

We now consider obstacle representations of \emph{co-bipartite}
graphs.
Recall that a graph is co-bipartite if its complement is bipartite.
Using this seemingly simple graph class, we settle an open problem
posed by Alpert et al.~\cite{akl-ong-DCG10} who asked if each
graph with obstacle number~1 has an outside-obstacle
representation. Namely, we provide an 11-vertex graph~$B_{11}$
(see Fig.~\ref{fig:K}) where not only is this not the case, but $B_{11}$ 
in fact has an inside-obstacle representation where the obstacle is the 
simplest possible shape, i.e., a triangle.\footnote{Note that for  
topologically closed obstacles, this obstacle could be a line 
segment.}  We also provide a smallest graph with obstacle number~2;
see the 8-vertex graph in Fig.~\ref{fig:B_8}. This improves on the
smallest previously known such graphs (e.g., the 10-vertex graphs of
Pach and Sar\i{\"o}z~\cite{ps-sglon-GC11} and of
Berman et al.~\cite{bcfghw-gong1-Manu16}) and shows that
Theorem~\ref{thm graphs order up to 7} 
is tight. 

\noindent\textit{\textbf{Properties of Outside-Obstacle Representations.}} We build on the easy observation (see Observation~\ref{obs:cliques_in_Out} 
below) that in every outside-obstacle representation of a graph, for every 
clique $Z$, the convex hull $\CH(Z)$ of the point set of $Z$ cannot be 
touched by the obstacle. In other words, the obstacle must occur outside of 
each such convex hull. Since we focus on co-bipartite graphs, this 
observation greatly restricts the ways one may realize an outside 
representation. Additionally, we will use this observation implicitly 
throughout this section whenever considering two cliques in a graph with an 
outside-obstacle representation. 

\begin{obs}\label{obs:cliques_in_Out}
If $G$ has an outside-obstacle representation $(P,C)$, then for every clique $Z \subseteq V(G)$, the convex hull $\CH(Z)$ of the points corresponding to $Z$ is disjoint from $C$, i.e., $C \cap \CH(Z) = \emptyset$.  
\end{obs}

For a graph $G$ containing two cliques $Z,Z' \subseteq V(G)$ and
outside-obstacle representation, consider the convex hulls $\CH(Z)$ and
$\CH(Z')$. We say that these convex hulls are \emph{$k$-crossing} when 
$\CH(Z) \setminus \CH(Z')$ consists of $k+1$ 
disjoint regions. Note that this condition is symmetric, i.e., when $\CH(Z) 
\setminus \CH(Z')$ consists of $r$ disjoint regions so does $\CH(Z') 
\setminus \CH(Z)$. 
We refer to these disjoint regions of the difference as the \emph{petals} 
of $Z$ ($Z'$ respectively). 

We now introduce a special 6-vertex graph $K^*_6$ which is used in the 
following technical lemma and our NP-hardness proof.  This graph is
the result of deleting a 3-edge matching from a 6-clique; see
Fig.~\ref{fig:co-bipartite-lemma(b)}.

\begin{lem}\label{lem:co-bipartite_outside-obstacles}
Let $G$ be a graph containing two cliques $Z,Z'$. For every outside-obstacle 
representation of $G$, the following properties hold.
\begin{enumerate}[(a)]
\item \label{enum:co-bip-priv} %
  If $\CH(Z)$ and $\CH(Z')$ are $t$-crossing, then every vertex in $Z$
  has at least $t-1$ neighbors in $Z'$ and vice versa.  That is, 
  if $Z$ contains a vertex with only $r$ neighbors in $Z'$,
  then $\CH(Z)$ and $\CH(Z')$ are at most $(r+1)$-crossing.
\item\label{enum:co-bip-K_6-match} If $G$ contains $K^*_6$ (with
  missing edges $z_1z'_1$, $z_2z'_2$, $z_3z'_3$; see
  Fig.~\ref{fig:co-bipartite-lemma(b)}) 
  as an induced subgraph, $\{z_1,z_2,z_3\} \subseteq Z$, and
  $\{z'_1,z'_2,z'_3\} \subseteq Z'$, then $\CH(\{z_1, z_2, z_3\})$ and
  $\CH(\{z_1',z_2',z_3'\})$ are at least 1-crossing. Furthermore,
  $\CH(Z)$ and $\CH(Z')$ are at least 1-crossing.
\item\label{enum:co-bip-4cycle} If $G$ contains a 4-cycle
  $z_1z_1'z_2'z_2$ as an induced subgraph, $\{z_1,z_2\}  \subseteq Z$,
  $\{z_1',z_2'\} \subseteq Z'$, $\CH(Z)$ and $\CH(Z')$ intersect,
  and~$z_1$ and~$z_2$ are contained in a petal~$Q^Z$ of~$Z$, then
  $z_1'$ and~$z_2'$ are contained in different petals of~$Z'$ which
  are both adjacent to~$Q^Z$.
  This implies that, if $\CH(Z)$ and $\CH(Z')$ are 1-crossing, then
  either~$z_1$ and~$z_2$ or $z_1'$ and~$z_2'$ are in different petals.
\end{enumerate}
\end{lem}

\begin{proof}
\ref{enum:co-bip-priv}
Suppose $\CH(Z)$ and $\CH(Z')$ are $t$-crossing for some $t \geq 2$.
Note that $|Z|,|Z'| \geq t+1$ since the convex hull of each must contain 
at least $t+1$ points. 
For $A \in \{Z,Z'\}$, let $Q^A_0, \dots, Q^A_{t}$ be the petals 
of $\CH(A)$ in clockwise order around $\CH(Z) \cap \CH(Z')$ where, for each 
$i \in \{0, \ldots, t\}$, $Q^Z_i$ is between $Q^{Z'}_i$ and $Q^{Z'}_{i+1}$ 
and all indices are considered modulo~$t+1$.  

Consider a vertex $v \in Z$ ($v \in Z'$ follows symmetrically). If $v$ is
in $\CH(Z) \cap \CH(Z')$, then we are done since $v$ sees every vertex in $Z'$
and $|Z'| \geq t+1$. So, suppose $v \in Q^Z_1$.
Consider the points $p_1=Q^{Z'}_1\cap Q^Z_0$ and $p_2=Q^{Z'}_2 \cap Q^Z_2$.
Define the subregion~$R$ (depicted as the grey region in 
Fig.~\ref{fig:co-bipartite-lemma(a)}) of $\CH(Z) \cup \CH(Z')$ 
whose boundary, in clockwise order, is formed by $\overline{p_1v}$,
$\overline{vp_2}$, and the polygonal chain from~$p_2$ to~$p_1$
along the boundary of $\CH(Z')$. 
Note that, for each $i \in \{0, 3, 4, \dots, t\}$, 
$Q^{Z'}_i \subset R$ and $R$ is convex, i.e., for every $u \in Q^{Z'}_i$, the 
line segment $\overline{vu}$ is contained in $\CH(Z) \cup \CH(Z')$. Thus, $v$ 
has at least $t-1$ neighbors in~$Z'$.  

\begin{figure}[tb]
 \begin{subfigure}[t]{0.194\textwidth}
  \centering     
  \includegraphics{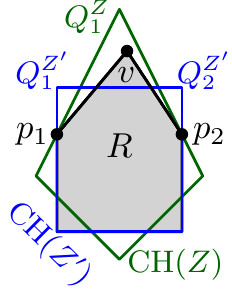}
  \caption{region $R$}
  \label{fig:co-bipartite-lemma(a)}
 \end{subfigure}
 \hfill
  \begin{subfigure}[t]{0.41\textwidth}
    \centering
    \includegraphics[page=2]{co-bip-4cycle}
    \caption{quadrilateral $z_1z_2z_2'z_1'$ is convex}
    \label{fig:co-bipartite-lemma(c)1}
  \end{subfigure}
  \hfill
  \begin{subfigure}[t]{0.37\textwidth}
    \centering
    \includegraphics[page=1]{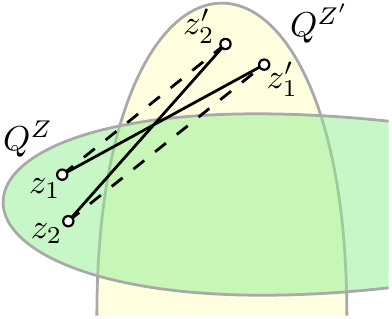}
    \caption{$\overline{z_1z_1'}$ and $\overline{z_2z_2'}$ intersect}
    \label{fig:co-bipartite-lemma(c)2}
  \end{subfigure}

  \caption{Aides for the proof of
    Lemma~\ref{lem:co-bipartite_outside-obstacles}.}
  \label{fig:co-bipartite-lemma}
\end{figure}

\ref{enum:co-bip-K_6-match} Consider the graph~$K^*_6$ as labeled in
Fig.~\ref{fig:co-bipartite-lemma(b)}.
We first show that the convex hulls of $X = \{z_1, z_2, z_3\}$ and 
$Y=\{z'_1,z'_2,z'_3\}$ are at least 1-crossing. 

Suppose that $\CH(X)$ and $\CH(Y)$ intersect but are 0-crossing. 
Since $|X|=|Y|=3$, a vertex in
$X \cup Y$ must be contained in $\CH(X) \cap \CH(Y)$.  Hence, this
vertex dominates $X \cup Y$, but $K^*_6$ doesn't have such a
vertex---a contradiction.

Now, suppose that $\CH(X)$ and $\CH(Y)$ are disjoint, and let $H = \CH(X \cup 
Y)$.  Since $\CH(X)$ and $\CH(Y)$ are disjoint, the boundary $\partial H$ of 
$H$ contains at most two line segments that connect a vertex of~$X$ to a 
vertex of $Y$, i.e., at most two non-edges of~$K^*_6$ occur on~$\partial H$. 
However, we will now see that every non-edge of $K^*_6$ must occur on $
\partial H$. Consider the line segment $\overline{z_1z'_1}$ and suppose it 
is not on $\partial H$. This means that there are vertices $u$ and $v$ of 
$K^*_6 \setminus \{z_1,z'_1\}$ where $u$ and $v$ occur on opposite sides of 
the line determined by $\overline{z_1z'_1}$. However, since $z_1z'_1$ is the 
only non-edge incident to either
$z_1$ or $z'_1$, the non-edge $z_1z'_1$ 
is enclosed by $\overline{uz_1}$, $\overline{z_1v}$, $\overline{vz'_1}$, $
\overline{z'_1u}$, which provides a contradiction. Thus, every non-edge 
must occur on $\partial H$, which contradicts the fact that at most two line 
segments spanning between $\CH(X)$ and $\CH(Y)$ can occur on $\partial H$. 

We now know that $\CH(X)$ and $\CH(Y)$ are at least 1-crossing. We use
this to observe that $\CH(Z)$ and $\CH(Z')$ must also be at least
1-crossing.  Clearly, if $\CH(Z)$ and $\CH(Z')$ are disjoint, this
contradicts $\CH(X)$ and $\CH(Y)$ being at least 1-crossing. So,
suppose that $\CH(Z)$ and $\CH(Z')$ intersect but are not
1-crossing.  Note that no vertex $v$ of $K^*_6$ is contained in
$\CH(Z) \cap \CH(Z')$ since otherwise $v$ would dominate to
$K^*_6$.  In particular, $X \subseteq \CH(Z) \setminus \CH(Z')$
and $Y \subseteq \CH(Z') \setminus \CH(Z)$. However, we again
would have $\CH(X)$ and $\CH(Y)$ being disjoint, i.e., 
a contradiction. Thus, $\CH(Z)$ and $\CH(Z')$ are at least 1-crossing.

\ref{enum:co-bip-4cycle} Suppose that $z'_1$ and $z'_2$ belong to the
same petal $Q^{Z'}$. This petal is adjacent to $Q^Z$, as otherwise $z_1$
would be visible to $z_2'$ (i.e., providing a contradiction). 
Now, if the quadrilateral $z_1z_2z_2'z_1'$ is convex, the
non-edge $z_1z_2'$ is not accessible from the outside (see
Fig.~\ref{fig:co-bipartite-lemma(c)1}).
If the quadrilateral $z_1z_2z_2'z_1'$ is non-convex, either a non-edge $z_1z_2'$ or a non-edge $z_2z_1'$ will not be accessible from the outside.
Thus, $\overline{z_1z_1'}$ and $\overline{z_2z_2'}$ intersect since $\CH(\{z_1,z_2\})$ and $\CH(\{z_1',z_2'\})$ are disjoint.
The edge $z_1z_1'$ together with the boundary of $\CH(Z) \cup \CH(Z')$ split the plane into at most two bounded and one unbounded region.
Then at least one of the non-edges $z_1z_2'$ and $z_1'z_2$ lies inside the union of the bounded regions.
This contradicts the fact that all non-edges should be accessible from the outside.
For example, in
Fig.~\ref{fig:co-bipartite-lemma(c)2}, the non-edge $z_1'z_2$ cannot
intersect any outside obstacle.
\end{proof}

\noindent\textit{\textbf{Inside- vs.\ Outside-Obstacle Graphs.}}
We now use Lemma~\ref{lem:co-bipartite_outside-obstacles} to show that
there is an 11-vertex graph (see $B_{11}$ in Fig.~\ref{fig:K})
that has 
an inside-obstacle representation but no outside-obstacle representation. 
This resolves an open question of Alpert et al.~\cite{akl-ong-DCG10}. We 
conjecture that, for any graph $G$ with at most 10 vertices, $\Obsin(G)=1$ 
implies $\Obsout(G)=1$.

  \begin{figure}[tb]
   \begin{subfigure}[t]{0.33\textwidth}
  \centering
  \includegraphics{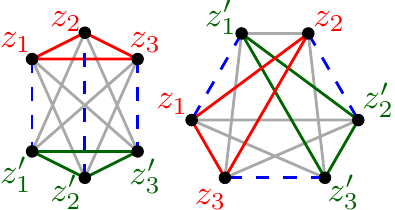}
  \caption{$K^*_6$ and its outside-obstacle representation}
  \label{fig:co-bipartite-lemma(b)}
  \end{subfigure}
  \hfill
    \begin{subfigure}[t]{.34\textwidth}
      \centering
      \includegraphics{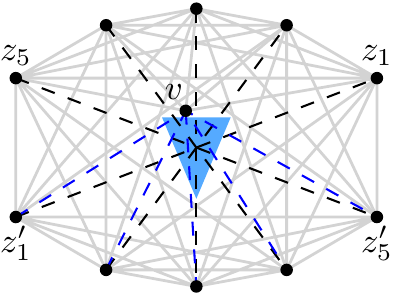}
      \caption{$B_{11}$ has $\Obsin(B_{11})=1$ (the obstacle is the
        blue triangle) but $\Obsout(B_{11})=2$}
      \label{fig:K}  
    \end{subfigure}
    \hfill
  \begin{subfigure}[t]{.26\textwidth}
    \centering
    \includegraphics{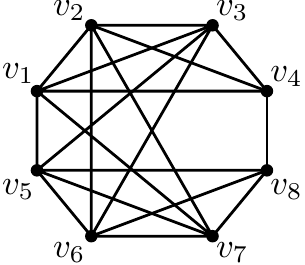}
    \caption{$B_8$ is a smallest graph of obstacle number~2}
    \label{fig:B_8}  
  \end{subfigure}
  \caption{Three small graphs: $K^*_6$, $B_{11}$ and $B_8$
  }
  \end{figure}

\begin{thm}\label{thm inside1 outside>1}
  There is an 11-vertex graph (e.g., $B_{11}$ in Fig.~\ref{fig:K})
\end{thm}
\begin{proof}
The 11-vertex co-bipartite graph $B_{11}$ is constructed as follows. We 
start with~$K_{10}$ on the vertices $z_1, \ldots, z_5, z'_1, 
\ldots, z'_5$. We then delete a 5-edge matching $\{z_iz'_i : i \in \{1,
\ldots, 5\}\}$ from $K_{10}$ to obtain $K^*_{10}$. Finally, we
obtain~$B_{11}$ by adding a vertex~$v$ adjacent to $z_1, \ldots,
z_5$. 
(Fig.~\ref{fig:K} shows an inside-obstacle representation
of~$B_{11}$ with a triangular obstacle.)

It remains to argue that $B_{11}$ has no outside-obstacle representation. 
Note that $B_{11}$ contains two cliques $Z = \{z_1, \ldots, z_5, v\}$ and 
$Z' = \{z'_1, \ldots, z'_5\}$. Furthermore, the vertex $v \in Z$ has no 
neighbors in $Z'$. Thus, by Lemma~\ref{lem:co-bipartite_outside-obstacles}
\ref{enum:co-bip-priv}, in any outside-obstacle representation, $\CH(Z)$ and 
$\CH(Z')$ are at most 1-crossing. Additionally, since each $z_i$ has a 
non-neighbor in $Z'$, no $z_i$ is contained in $\CH(Z) \cap \CH(Z')$. In 
particular, since $Z$ has only two petals, there are three $z_i$'s, say 
$z_1,z_2,z_3$, that are contained in a single petal of~$Z$.
Now note that $K^*_6$ is 
the subgraph of $B_{11}$ induced by $\{z_1,z_2,z_3,z'_1,z'_2,z'_3\}$.
Since $z_1,z_2,z_3$ are contained in a petal of $Z$,  
$\CH(\{z_1,z_2,z_3\})$ and $\CH(\{z'_1,z'_2,z'_3\})$ are disjoint,
contradicting Lemma~\ref{lem:co-bipartite_outside-obstacles}
\ref{enum:co-bip-K_6-match}. Thus, $B_{11}$ has outside-obstacle 
number~2. 
\end{proof}

Note that a graph with an inside-obstacle representation is either a 
clique or contains a cycle since an inside obstacle cannot (by
definition) pierce the convex hull of the point set\footnote{%
In Appendix~\appref{appendixD}, we show that $K_{2,3}$ is the smallest
graph with a cycle and an outside-obstacle 
representation but no inside-obstacle representation.}. 
Thus, by Theorem~\ref{thm inside1 outside>1} and this fact, we have the
following. 

\begin{thm}\label{thm:inside||outside}
  The classes of inside-obstacle representable graphs and
  outside-obstacle representable graphs are incomparable.
\end{thm}

\noindent\textit{\textbf{Obstacle Number~2.}}
We present an 8-vertex graph (see $B_8$ in Fig.~\ref{fig:B_8}) with 
obstacle number~2. To prove this result, we first apply
Lemma~\ref{lem:co-bipartite_outside-obstacles} to show that $B_8$ has
no outside-obstacle representation.
In Lemma~\ref{lem:8-vertex-no-inside-obstacle} (proven in Appendix~\ref{appendixB}), we 
demonstrate that $B_8$ also has no inside-obstacle representation.
In particular, these lemmas together with Theorem~\ref{thm graphs order up to 7}
provide the following theorem.
 
\begin{thm}\label{thm smallest graph}
  The smallest graphs without a single-obstacle representation
  have eight vertices, e.g., the co-bipartite graph~$B_8$ in
  Fig.~\ref{fig:B_8}.
\end{thm}
\begin{proof}
The graph $B_8$ has $8$ vertices $v_1, \ldots, v_8$. It has precisely the following set of non-edges: $v_1v_6$, $v_2v_5$, $v_3v_7$, $v_4v_5$, $v_4v_6$, $v_4v_7$, $v_8v_1$, $v_8v_2$, $v_8v_3$. Note that the subgraph induced by $\{v_1, v_2, v_3, v_5, v_6, v_7\}$ is a $K^*_6$. Further, note that $Z = \{v_1, v_2, v_3, v_4\}$ and $Z' = \{v_5, v_6, v_7, v_8\}$ are cliques. 
  
Suppose (for a contradiction) $B_8$ has an outside-obstacle representation. 
By Lemma~\ref{lem:co-bipartite_outside-obstacles} \ref{enum:co-bip-K_6-match}, 
$\CH(Z)$ and $\CH(Z')$ are at least 1-crossing. Additionally, since 
$v_4$ has only one neighbor in $Z'$, we know that $\CH(Z)$ and $\CH(Z')$ are 
at most 2-crossing. We will consider these two cases separately.
Let $Q^Z_0$, $Q^Z_1$, $Q^Z_2$ be the petals of $Z$ and 
$Q^{Z'}_0$, $Q^{Z'}_1$, $Q^{Z'}_2$ be the petals of $Z'$ where the cyclic 
order of the petals around $\CH(Z) \cap \CH(Z')$ is $Q^{Z'}_0,Q^Z_0$, 
$Q^{Z'}_1$, $Q^Z_1$, $Q^{Z'}_2$, $Q^Z_2$. Note that every vertex is contained in one of the petals. 

\medskip
\noindent\textbf{Case 1:} $\CH(Z)$ and $\CH(Z')$ are 2-crossing.
Suppose $v_4 \in Q^Z_0$. Since $v_8$ is the only neighbor of $v_4$ in $Z'$, 
we must have $v_8 \in Q^{Z'}_2$, and 
now the only vertex in $Q^Z_0$ is $v_4$ and the only vertex in $Q^{Z'}_2$ is 
$v_8$. However, we now have $\{v_1,v_2,v_3\} \subset Q^Z_1 \cup Q^Z_2$ 
and $\{v_5,v_6,v_7\} \subset Q^{Z'}_0 \cup Q^{Z'}_1$, i.e., 
$\CH(\{v_1,v_2,v_3\})$ and $\CH(\{v_5,v_6,v_7\}$ are disjoint,
contradicting Lemma~\ref{lem:co-bipartite_outside-obstacles}~\ref{enum:co-bip-K_6-match}.

\medskip
\noindent\textbf{Case 2:} $\CH(Z)$ and $\CH(Z')$ are 1-crossing.
Note that $v_1$, $v_2$, and $v_3$ cannot belong to the same petal 
(otherwise, we would contradict Lemma~\ref{lem:co-bipartite_outside-obstacles}~\ref{enum:co-bip-K_6-match}). Similarly, $v_5$, $v_6$, and $v_7$ 
cannot belong to the same petal. 
Thus, without loss of generality, we have 
$v_1$ and $v_2$ in $Q^Z_0$, 
$v_3$ in $Q^{Z}_1$, 
$v_5$ and $v_7$ in $Q^{Z'}_0$, and 
$v_6$ in $Q^{Z'}_1$.
When $v_4$ is in $Q^Z_0$ and $v_8$ is in $Q^{Z'}_0$, 
the induced 4-cycle $v_4v_2v_7v_8$ contradicts
Lemma~\ref{lem:co-bipartite_outside-obstacles}~\ref{enum:co-bip-4cycle}. 
Similarly, when $v_4$ is in $Q^Z_0$ and $v_8$ is in $Q^{Z'}_1$, we use the induced 4-cycle $v_4v_2v_6v_8$;
when $v_4$ is in $Q^Z_1$ and $v_8$ is in $Q^{Z'}_0$, we use the induced 4-cycle $v_4v_3v_5v_8$;
and when $v_4$ is in $Q^Z_1$ and $v_8$ is in $Q^{Z'}_1$, we use the induced 4-cycle $v_4v_3v_6v_8$.

It remains to show that $B_8$ has no inside-obstacle representation 
(formalized in Lemma~\ref{lem:8-vertex-no-inside-obstacle} below). 
This is proven in Appendix~\appref{appendixB}.
\end{proof}

\newcommand{\LemEightVertexNoInsideObstacle}{%
The graph $B_8$ in Fig.~\ref{fig:B_8} has no inside-obstacle 
representation.}
\wormhole{lem:8-vertex-no-inside-obstacle}
\begin{lem}\label{lem:8-vertex-no-inside-obstacle}
\LemEightVertexNoInsideObstacle
\end{lem}

\section{NP-Hardness}
\label{sec:hardness}

In this section, we show that the single-obstacle, outside-obstacle,
inside-obstacle graph sandwich problems as well as the simple-polygon
visibility graph sandwich problem are all NP-hard.  Note that the
complexity of the obstacle graph sandwich problem yields an upper
bound for the complexity of our (simpler) recognition problem.

\begin{thm}\label{hardness}
  The outside-obstacle graph sandwich problem is NP-hard.  In other
  words, given two graphs~$G$ and~$H$ with the same vertex set and $G
  \subseteq H$, it is NP-hard to decide whether there is a graph~$K$
  such that $G \subseteq K \subseteq H$ and $\Obsout(K)=1$.
  This holds even if $G$ and $H$ are co-bipartite.
\end{thm}

\begin{proof}
  We reduce from \textsc{MonotoneNotAllEqual3Sat}, which is
  NP-hard~\cite{s-csp-STOC78}.  In this version of \textsc{3Sat}, all
  literals are positive, and the task is to decide whether the given
  \textsc{3Sat} formula~$\varphi$ admits a truth assignment such that
  in each clause at least one and at most two variables are true.
  
  Given~$\varphi$, we build a graph~$G_\varphi$ with edges, non-edges
  and ``maybe''-edges such that $\varphi$ is a yes-instance if and only
  if~$G_\varphi$ has a subgraph that has an outside-obstacle
  representation and contains all edges, no non-edges
  and an arbitrary subset of the maybe-edges.  (In other words, the
  set of edges of~$G_\varphi$ yields~$G$ in the statement of the
  theorem, and the set of edges and maybe-edges yields~$H$.)
  Let $\{v_1,\dots,v_n\}$ be the set of variables, and let
  $\{C_1,\dots,C_m\}$ be the set of clauses in~$\varphi$.  For
  $i=1,\dots,n$, let~$v_{ij}$ be the $j$-th occurrence of $v_i$ in
  $\varphi$.

  Now we can construct~$G_\varphi$.  For each variable, we introduce a
  \emph{variable vertex} (of the same name).  These $n$ vertices form
  a clique.  For each occurrence~$v_{ij}$ of a variable~$v_i$
  in~$\varphi$, we introduce an \emph{occurrence vertex} (of the same
  name).  These $3m$ vertices also form a clique.  In order to
  restrict how the two cliques intersect, we add
  to~$G_\varphi$ a copy of $K_6^*$ labeled as in
  Fig.~\ref{fig:co-bipartite-lemma(b)}; vertices $z_1,z_2,z_3$
  participate in the occurrence-vertex clique, whereas vertices
  $z_1',z_2',z_3'$ participate in the variable-vertex clique.
  We add one more vertex~$u$ to the occurrence-vertex clique.
  The special vertex~$u$ is adjacent to $z_3'$ and has non-edges to
  all other vertices in the variable-vertex clique.
  The edge set of~$G_\varphi$ depends on~$\varphi$ as follows.
  Each variable vertex~$v_i$ has
  \begin{itemize}[topsep=0pt]
  \item an edge to any occurrence vertex~$v_{ij}$,
  \item a non-edge to any occurrence vertex~$v_{k\ell}$ that
    represents an occurrence of a variable~$v_k$ that co-occurs
    with~$v_i$ in some clause of~$\varphi$,
  \item a maybe-edge to any other occurrence vertex.
  \end{itemize}
\medskip

  Next, we show how to use a feasible truth assignment of~$\varphi$ to
  lay out~$G_\varphi$ so that all its non-edges are accessible from
  the outside.  We place the vertices on the boundary of two
  intersecting rectangles, one for each clique.  Given these
  positions, we show that all non-edges intersect the outer face of
  the union of the edges.  Finally, we bend the sides of the
  rectangles slightly into very flat circular arcs such that all of
  the previous (non-) visibilities remain and the vertices are in
  general position.

  We take two axis-aligned rectangles~$R_1$ and~$R_2$ that intersect
  as a cross; see Fig.~\ref{fig:hardness-sascha}.  Let
  $X_1,X_2,X_3,X_4$ be the corners of $R_1 \cap R_2$ in clockwise
  order, starting in the lower left corner.  We place the
  variable vertices on the boundary of the ``wide'' rectangle~$R_1$:
  the vertices $v_1,\dots,v_p$ of the true variables are equally
  spaced from top to bottom on a segment on the left side, similarly
  the vertices $v_{p+1},\dots,v_n$ of the false variables go to a
  segment on the right side.  (In Fig.~\ref{fig:hardness-sascha}(b),
  $p=3$.)  The two vertical segments are chosen such that they ``see''
  four disjoint horizontal segments on the top and bottom edge
  of~$R_2$; refer to Fig.~\ref{fig:hardness-sascha}(a) for the
  positions of the six segments in total.
  
  \begin{figure}[tb]
    \centering
    \includegraphics[page=4,scale=.75]{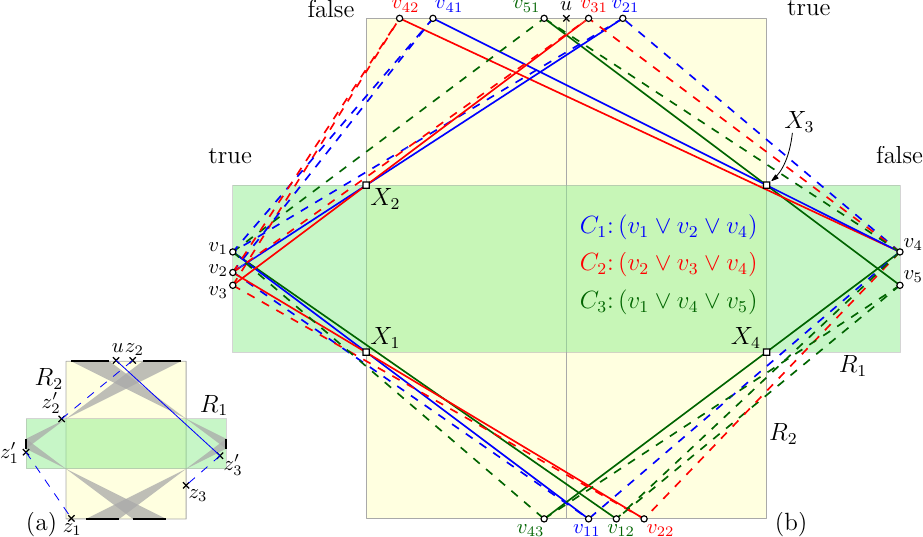}
    \caption{NP-hardness: maybe-edges and the two cliques
      are not drawn.}
    \label{fig:hardness-sascha}
  \end{figure}

  In each clause, we sort the variables in increasing order of index.
  We place the occurrence vertices on the horizontal segments
  of~$R_2$.  For a true variable~$v_i$ (such as~$v_2$ in
  Fig.~\ref{fig:hardness-sascha}(b) the first occurrence
  vertex~$v_{i1}$ has two potential locations; the bottom location is
  where the ray from~$v_i$ through~$X_1$ hits the bottom right
  segment, the top location is where the ray from~$v_i$ trough~$X_2$
  hits the top right segment.  We place~$v_{i1}$ to its bottom or top
  location depending on whether $v_{i1}$ is the first or second
  occurrence of a true variable in its clause, respectively.
  (Remember that within each clause, at most two variables are true
  and at most two are false.)  Occurrence vertices~$v_{i2}$ etc.\ go
  between the top or bottom locations of~$v_{i1}$ and $v_{i+1,1}$,
  again depending on whether they are the first or second occurrence
  of a true variable in their respective clauses.  (E.g., in
  Fig.~\ref{fig:hardness-sascha}(b)), $v_{21}$ goes to the top,
  whereas $v_{22}$ goes to the bottom.)

  The special vertex~$u$ is placed in the center of the top edge
  of~$R_2$; hence, it is not visible from any variable vertex; see
  Fig.~\ref{fig:hardness-sascha}(a).  The vertices of~$K_6^*$ can be
  placed such that~$u$ sees only~$z_3'$, but neither $z_1'$
  nor~$z_2'$; see Fig.~\ref{fig:hardness-sascha}(b)).

  By construction, all edges are inside~$R_1 \cup R_2$.  It remains to
  show that all non-edges (dashed in
  Fig.~\ref{fig:hardness-sascha}(b)) go through the complement of
  $R_1 \cup R_2$.  This is due to the order of the variable vertices
  and the occurrence vertices along the boundary of $R_1 \cup R_2$ and
  due to the order of the variables in each clause.  Suppose that a
  variable vertex~$v_i$ has a non-edge with occurrence
  vertex~$v_{k\ell}$.  This means that there is an occurrence~$v_{ij}$
  of~$v_i$ in the same clause as~$v_{k\ell}$.  If~$v_i$ and~$v_k$ have
  different truth values, then $v_i$ cannot see~$v_{k\ell}$; refer to
  Fig.~\ref{fig:hardness-sascha}(a).  So assume that both are true and
  that $i<k$.  But then $v_i$ lies above~$v_k$ on the left segment
  of~$R_1$, and $v_{ij}$ lies to the left of~$v_{k\ell}$ on the bottom
  right segment of~$R_2$.  Hence, $v_i$ cannot see~$v_{k\ell}$.

  It remains to show that an outside-obstacle representation
  of~$G_\varphi$ yields a feasible truth assignment for~$\varphi$.  By
  Lemmas~\ref{lem:co-bipartite_outside-obstacles}\ref{enum:co-bip-priv}
  and~\ref{enum:co-bip-K_6-match}, we know that the convex hulls of
  the two cliques are at least 1-crossing due to the presence of~$K_6^*$ 
  and at most 2-crossing due to $u$. To see that these hulls are 
  exactly 1-crossing, we suppose that $G_\varphi$ has a 2-crossing drawing
  for a contradiction. 
  Consider the subgraph $H$ induced by $u$ and the first
  clause $C_1 = \{v_{i},v_{j},v_{k}\}$, of $\varphi$ i.e., $H = G[\{u, v_{i},
  v_{j},v_{k}, v_{i1},v_{j1},v_{k1}\}]$. Let $Q_u$ be the petal containing
  $u$. Since the only neighbor of $u$ in the variable-vertex clique is $z'_3$,
  no other variable vertices belong to the petal opposite $Q_u$. 
  Thus, two of $\{v_{i},v_{j},v_{k}\}$, say $v_i$ and $v_j$,
  occur in one petal $Q'_1$ adjacent to $Q_u$, and $v_k$ occurs in the other 
  petal $Q'_2$ which is adjacent to $Q_u$. Notice that each of $v_{i1},v_{j1},
  v_{k1}$ cannot belong to the petal opposite $Q'_1$ since this
  would make it adjacent to both $v_{i}$ and $v_{j}$. Similarly, no 
  neither $v_{i1}$ nor $v_{j1}$ can occur in the petal opposite $Q'_2$
  since it would then be adjacent to~$v_k$. Thus, $v_{i1}$ and $v_{j1}$ belong
  to the same petal and this petal is adjacent to $Q'_1$. However, this contradicts
  Lemma~\ref{lem:co-bipartite_outside-obstacles}\ref{enum:co-bip-4cycle} since 
  $\{v_i,v_j,v_{i1},v_{j1}\}$ induces a 4-cycle. 
    
  Now, since the convex hulls are exactly 1-crossing, we have
  two groups (petals) of vertices in each of the two cliques.
  \WLOG the variable-vertex clique is divided into a left and a right
  group, and the occurrence-vertex clique is divided into a top and a
  bottom group.  We set those variables to true whose vertices lie on
  the left, the rest to false.

  Now suppose that the three variables $v_1$, $v_2$, and~$v_3$ of
  clause~$C_1$ lie in the same group, say, on the left.  Then two of
  their occurrence vertices (say~$v_{11}$ and~$v_{21}$) lie in the
  same group, say, in the top group.  Since $v_1v_{21}$ and
  $v_2v_{11}$ are non-edges, $v_1v_{11}v_{21}v_2$ is an induced
  4-cycle.  Now
  Lemma~\ref{lem:co-bipartite_outside-obstacles}\ref{enum:co-bip-4cycle},
  yields the desired contradiction.  Hence, no three variable
  vertices in a clause can be in the same (left or right)
  group. Therefore, our truth assignment is indeed feasible.  This
  completes the NP-hardness proof.
\end{proof}

To show hardness for the simple-polygon
visibility graph sandwich problem,  we must make sure that any vertex
of the obstacle is also a vertex of the graph.  It suffices to add 
$X_1,X_2,X_3,X_4$ 
as vertices to~$G_\varphi$ that lie in both cliques.

\begin{thm}\label{thm:visibility-graph}
  The simple-polygon visibility graph sandwich problem is NP-hard.  
  In other words, given two graphs~$G$ and~$H$ with the same vertex
  set and $G \subseteq H$, it is NP-hard to decide whether there is a
  graph~$K$ and a polygon~$\Pi$ such that $G \subseteq K \subseteq H$
  and $K=G_\Pi(V(\Pi))$.
  This holds even if $G$ and $H$ are co-bipartite.
\end{thm}

We can also use the NP-hardness of the outside-obstacle 
sandwich problem to show NP-hardness for both the single-obstacle 
sandwich problem and the inside-obstacle sandwich problem. 
The idea is simply to combine a given graph $G$ with a graph such as $B_{11}$ which
has outside-obstacle number greater than one, but inside-obstacle number one. 
The combined graph would then have inside-obstacle number one if and only if
the graph $G$ has outside-obstacle number one. 
The details of this are given in Appendix~\appref{appendixC}.

\bibliographystyle{splncs03}
\bibliography{abbrv,obstacles}

\begin{thebibliography}{10}
\providecommand{\url}[1]{\texttt{#1}}
\providecommand{\urlprefix}{URL }

\bibitem{akl-ong-DCG10}
Alpert, H., Koch, C., Laison, J.D.: Obstacle numbers of graphs. Discrete
  Comput. Geom.  44(1),  223--244 (2009),
  \url{http://dx.doi.org/10.1007/s00454-009-9233-8}

\bibitem{bcv-dgusno-GD15}
Balko, M., Cibulka, J., Valtr, P.: Drawing graphs using a small number of
  obstacles. In: Di~Giacomo, E., Lubiw, A. (eds.) Proc. 23rd Int. Symp. Graph
  Drawing (GD'15). Lect. Notes Comput. Sci., vol. 9411, pp. 360--372.
  Springer-Verlag (2015)

\bibitem{bcfghw-gong1-Manu16}
Berman, L.W., Chappell, G.G., Faudree, J.R., Gimbel, J., Hartman, C., Williams,
  G.I.: Graphs with obstacle number greater than one. Arxiv report
  \href{http://arxiv.org/abs/1606.03782}{arxiv.org/abs/1606.03782} (2016)

\bibitem{ch-rcpvg-SoCG15}
Cardinal, J., Hoffmann, U.: Recognition and complexity of point visibility
  graphs. In: Arge, L., Pach, J. (eds.) Proc. 31st Int. Symp. Comput. Geom.
  (SoCG'15). LIPIcs, vol.~34, pp. 171--185. Schloss Dagstuhl~-- Leibniz-Zentrum
  f\"ur Informatik (2015)

\bibitem{dm-on-EJC15}
Dujmovi{\'c}, V., Morin, P.: On obstacle numbers. Electr. J. Combin.  33(3),
  paper \#P3.1, 7~pages (2015), see also
  \href{http://arxiv.org/abs/1308.4321}{arxiv.org/abs/1308.4321}

\bibitem{fss-conog-30EGGT}
Fulek, R., Saeedi, N., Sar\i{\"o}z, D.: Convex obstacle numbers of outerplanar
  graphs and bipartite permutation graphs. In: Thirty Essays on Geometric Graph
  Theory, pp. 249--261. Springer, New York (2013)

\bibitem{gg-upvgp-arXiv12}
Ghosh, S.K., Goswami, P.P.: Unsolved problems in visibility graphs of points,
  segments and polygons. Arxiv report
  \href{https://arxiv.org/abs/1012.5187v4}{arxiv.org/abs/1012.5187v4} (2012),
  \url{http://arxiv.org/abs/1012.5187}

\bibitem{gks-gsp-JA95}
Golumbic, M.C., Kaplan, H., Shamir, R.: Graph sandwich problems. J. Algorithms
  19(3),  449--473 (1995)

\bibitem{js-rpslg-CCCG14}
Johnson, M.P., Sar\i\"oz, D.: Representing a planar straight-line graph using
  few obstacles. In: Proc. 26th Canadian Conf. Comput. Geom. (CCCG'14). pp.
  95--99 (2014), \url{http://www.cccg.ca/proceedings/2014/papers/paper14.pdf}

\bibitem{kkr-gpoor-arXiv13}
Koch, A., Krug, M., Rutter, I.: Graphs with plane outside-obstacle
  representations. Arxiv report
  \href{http://arxiv.org/abs/1306.2978}{arxiv.org/abs/1306.2978} (2013)

\bibitem{mpp-lbong-EJC12}
Mukkamala, P., Pach, J., P\'alv\"olgyi, D.: Lower bounds on the obstacle number
  of graphs. Electr. J. Combin.  19(2),  paper \#P32, 8~pages (2012),
  \url{http://www.combinatorics.org/ojs/index.php/eljc/article/view/v19i2p32}

\bibitem{mps-glon-WG10}
Mukkamala, P., Pach, J., Sar\i\"oz, D.: Graphs with large obstacle numbers. In:
  Thilikos, D.M. (ed.) Proc. Conf. Graph-Theoretic Concepts Comput. Sci.
  (WG'10). Lect. Notes Comput. Sci., vol. 6410, pp. 292--303. Springer-Verlag
  (2010), \url{http://dx.doi.org/10.1007/978-3-642-16926-7_27}

\bibitem{ps-sglon-GC11}
Pach, J., Sar\i{\"o}z, D.: On the structure of graphs with low obstacle number.
  Graphs and Combinatorics  27(3),  465--473 (2011),
  \url{http://dx.doi.org/10.1007/s00373-011-1027-0}

\bibitem{s-csp-STOC78}
Schaefer, T.J.: The complexity of satisfiability problems. In: Proc. 10th Annu.
  {ACM} Symp. Theory Comput. (STOC'78). pp. 216--226 (1978),
  \url{http://dx.doi.org/10.1145/800133.804350}

\end{thebibliography}

\newpage
\appendix

\noindent{\Large\bfseries Appendix}

\section{Missing Proofs of Section~\ref{sec:small_circ}}
\label{appendixA}

This appendix contains the omitted proof of lemmas for Theorem~\ref{thm graphs circumference up to 6} and the missing part of proof for Theorem~\ref{thm graphs circumference up to 6} and Theorem~\ref{thm graphs order up to 7}.

\begin{backInTime}{lem merge graphs}
\begin{lem}
  \LemMergeGraphsText
\end{lem}
\end{backInTime}
\begin{proof}
  \begin{figure}[b]
    \centering
    \includegraphics{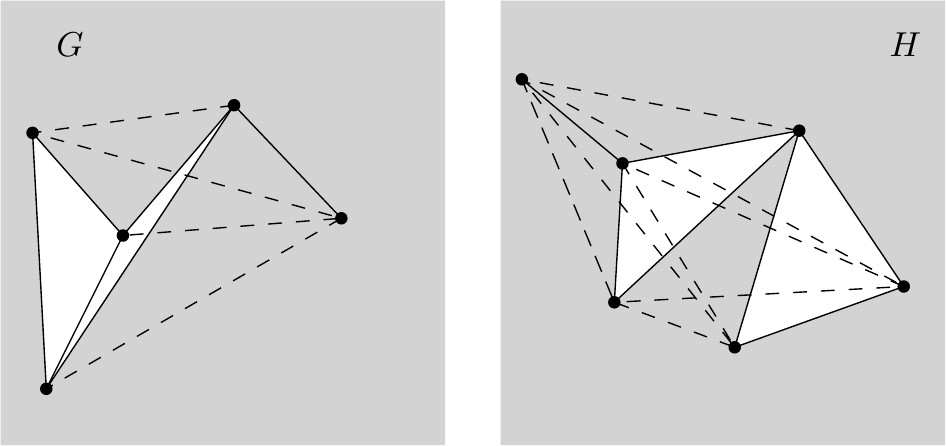}
    \caption{Grey (open) regions are outside obstacles of $G$ and $H$ respectively}
    \label{fig:lem1}
  \end{figure}
Fix outside-obstacle representations for $G$ and $H$.
We can assume $G$ lies inside $(0,1)\times(0,1)$ and $H$ lies inside $(1,2)\times(0,1)$ by scaling.
Let $C_G$ be an outside obstacle for $G$ and $C_H$ for $H$.
We can also assume that $\partial ([0,1]\times [0,1]) \subset \partial C_G$ and $\partial ([1,2]\times [0,1]) \subset \partial C_H$.
Take $C=C_G \cup C_H \cup (\{1\}\times [0,1])$.
We claim that $C$ is an obstacle for $G \cup H$.
Let $v$ be a vertex of $G$ and $u$ a vertex of $H$.
Since $\overline{vu}$ intersects with $\{1\}\times [0,1]$ and $\{1\}\times [0,1] \subset C$, we indeed have $uv$ as a non-edge.
See Fig.~\ref{fig:lem1} for an example illustration.
\end{proof}

\begin{backInTime}{lem identifying graphs}
\begin{lem}
  \LemIdentifyingGraphsText
\end{lem}
\end{backInTime}
\begin{proof}
  \begin{figure}[b]
    \centering
    \includegraphics{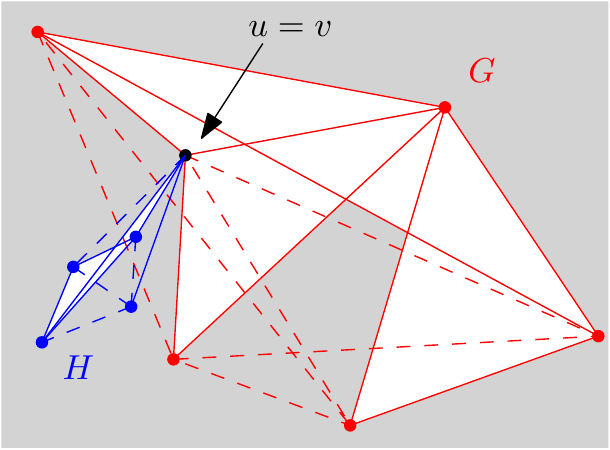}
    \caption{$G$ is drawn with red color; $H$ with blue. Grey (open) region is an outside obstacle of $K$.}
    \label{fig:lem2}
  \end{figure}
Fix an exposed outside-obstacle representation of $H$ such that $v$ is on the convex hull of $V(H)$.
Let $C=z_1z_2\dots z_mz_1$ be the one of boundaries of the outside obstacle such that the obstacle lies in the unbounded component of $\mathbb{R}^2 \setminus C$ and all vertices of $H$ lies in the bounded component.
Let $l,m$ be two rays starting in $v$ where all vertices of $H$ are between $l$ and $m$.
Since $v$ is on the convex hull of $V(H)$, the angle between $l$ and $m$ is less than $\pi$.
\WLOG $v$ is placed at the origin and the ray from $v$ to $(1,0)$ is between $l$ and $m$.

We first show that squashing and shrinking with respect to $v$ preserves the structure of outside-obstacle representation.
More precisely, for $s, t > 0$, let $T_{s,t}$ be a transformation mapping a point $(x, y)$ to $(sx, ty)$.
We show that the point set with the obstacle obtained by transforming each vertex of $H$ and its outside obstacle by $T_{s,t}$ is still an exposed outside-obstacle representation of $H$ where $v$ is on the convex hull of $V(H)$.
Let $C'=z_1'\dots z_m'z_1' =T_{s,t}(z_1)\dots T_{s,t}(z_m)T_{s,t}(z_1)$ and $a'=T_{s,t}(a)$ for $a\in V(H)$ and for simplicity.
Let $a, b$ be vertices of $H$.
If $a\sim b$, suppose $\overline{a'b'}$ intersects with the obstacle for contradiction.
It is clear that if $\overline{a'b'}$ intersects with $\overline{z_i'z_{i+1}'}$ then $\overline{ab}$ intersects with $\overline{z_iz_{i+1}}$, contradicting the fact that $ab$ should not intersect $C$.
Similarly, we can also show that if $a\not\sim b$ then $\overline{a'b'}$ intersects with $C'$, that every transformed vertex is exposed to the outside, and that $v'$ is on the convex hull of the transformed point set.

Consequently, we can assume that $H$ has an exposed outside-obstacle representation
where all vertices are contained in an open circular sector centered at
$v$ whose radius and angle are arbitrarily small such
that the boundary of the outside obstacle includes the boundary of the arc.
Fix an exposed outside-obstacle representation of $G$ such that the outside
obstacle is maximal (i.e., the outside obstacle is the unbounded
component of the complement of the visibility drawing). 
Since $u$ is exposed to the outside, we can find an arc sector $A$ of radius $r$ and angle $\theta$ centered at $u$, completely lying inside the outside obstacle except $u$.
We replace $A$ with above obstacle representation of $H$ while identifying $u$ and $v$.
We claim that it is an outside-obstacle representation of $K$.
Since all edges/non-edges of $H$ lie inside $A$, they are properly represented with the new obstacle.
Since new obstacle is a subset of obstacle of $G$, all edges of $G$ don't intersect with the obstacle.
For non-edges of $G$, if they didn't intersect $A$, they would still intersect the obstacle.
Otherwise, since $\partial A$ is contained in the obstacle, they would also intersect the obstacle.
For a vertex of $G$ except $u$ and a vertex of $H$ except $v$, they are always non-adjacent and it is properly represented since the line segment connecting them intersects with $\partial A$.
Lastly, every vertex of $H$ is exposed to the outside by previous paragraph and every vertex of $G$ except $u$ is also exposed to the outside because only the subset of $A$ is altered in the obstacle.
See Fig.~\ref{fig:lem2} for an example illustration.
\end{proof}

\begin{backInTime}{lem nonadjacent twin}
\begin{lem}
  \LemNonadjacentTwinText
\end{lem}
\end{backInTime}
\begin{proof}

  \begin{figure}[b]
	\centering
    \includegraphics{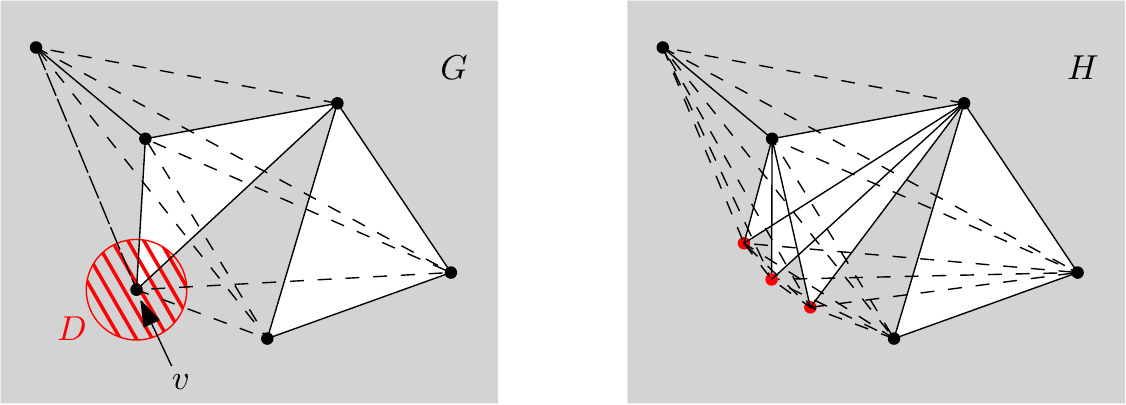}
    \caption{Grey (open) regions are outside obstacles of $G$ and $H$ respectively; $A\cup\{v\}$ are denoted by red color in $H$.}
    \label{fig:lem3}
  \end{figure}
  
Fix an exposed outside-obstacle representation of $G$ such that $v$ is on the convex hull of $V(G)$. We choose $\epsilon$ small enough so that a disk $D$ of radius $\epsilon$ centered at $v$ doesn't contain any other vertices.
We also want $\epsilon$ to be small enough so that the outside-obstacle representation where a point for $v$ is replaced by any point in $D$ is still a valid outside-obstacle representation for $G$.
This guarantees that adding $A$ inside $D$ results in a valid outside-obstacle representation for $H$.

More precisely, let $p, q$ be two intersection points between the convex hull of $V(G)$ and $D$.
We make a slightly bended outwards (for general position assumption) segment $C$ connecting $p$ and $q$.
We then place $v$ and vertices of $A$ on $C$, say evenly.
Since the only part of $A$ is altered from the obstacle, all vertices of $H$ except $A\cup\{v\}$ are exposed to the outside. Since replacing the polygonal curve $pvq$ with $C$ from the convex hull of $V(G)$ still yields a convex region, all vertices of $A\cup\{v\}$ are on the convex hull of $V(H)$ and exposed to the outside.
See Fig.~\ref{fig:lem3} for an example illustration.
\end{proof}

\begin{backInTime}{thm graphs circumference up to 6}
\begin{thm}
  \ThmCircumferenceUpToSixText
\end{thm}
\end{backInTime}
\begin{proof}
 
  \begin{figure}
    \newcommand{\cSixRadius}{1.66}
    \newcounter{x}
    \newcommand{\cSixHexagon}[1]{
      \setcounter{x}{0}
      \foreach \i in {#1} {
        \coordinate (\i) at ({\thex*60+30}:\cSixRadius);
        \node[anchor={\thex*60+30+180}] at (\i) {$v_\i$};
        \stepcounter{x}
      }
    }
    \newcommand{\cSixPentagon}[1]{
      \setcounter{x}{0}
      \foreach \i in {#1} {
        \ifthenelse{\thex=4}{
          \coordinate (\i) at ({\thex*60+30}:0.3*\cSixRadius);
        }{
          \coordinate (\i) at ({\thex*60+30}:\cSixRadius);
        }
        \node[anchor={\thex*60+30+180}] at (\i) {$v_\i$};
        \stepcounter{x}
      }
    }
    \newcommand{\cSixDrawPoints}{
      \foreach \i in {1,...,6} {
        \fill (\i) {} circle(2pt);
      }
    }
    \tikzset{
      edge/.style={draw,black,line width=1pt},
      nonedge/.style={draw,black,dashed,line width=1pt},
      ruleedge/.style={draw,green!80!black,dashed,line width=1pt},
      optionaledge/.style={draw,blue,dashed,line width=1pt}
    }

    \begin{subfigure}{0.32\textwidth}
      \begin{tikzpicture}
        \cSixHexagon{4,1,6,3,2,5}
        \draw[edge] (1) -- (2) -- (3) -- (4) -- (5) -- (6) -- (1);
        \draw[edge] (1) -- (3)
          (1) -- (5)
          (3) -- (5)
          (2) -- (4)
          (2) -- (6)
          (4) -- (6);
        \draw[optionaledge] (1) -- (4)
          (2) -- (5)
          (3) -- (6);
        \cSixDrawPoints
      \end{tikzpicture}
      \caption{}
      \label{fig:c6:b1}
    \end{subfigure}
    \hfill
    \begin{subfigure}{0.32\textwidth}
      \begin{tikzpicture}
        \cSixHexagon{5,3,1,6,2,4}
        \draw[edge] (1) -- (2) -- (3) -- (4) -- (5) -- (6) -- (1);
        \draw[edge] (3) -- (6)
          (1) -- (4);
        \draw[nonedge] (2) -- (4)
          (1) -- (3);
        \draw[optionaledge] (3) -- (5)
          (1) -- (5)
          (2) -- (5)
          (2) -- (6)
          (4) -- (6);
        \cSixDrawPoints
      \end{tikzpicture}
      \caption{}
      \label{fig:c6:a2}
    \end{subfigure}
    \hfill
    \begin{subfigure}{0.32\textwidth}
      \begin{tikzpicture}
        \cSixHexagon{1,3,5,2,6,4}
        \draw[edge] (1) -- (2) -- (3) -- (4) -- (5) -- (6) -- (1);
        \draw[edge] (3) -- (6)
          (2) -- (4);
        \draw[nonedge] (1) -- (3);
        \draw[optionaledge] (3) -- (5)
          (1) -- (4)
          (1) -- (5)
          (2) -- (5)
          (2) -- (6)
          (4) -- (6);
        \cSixDrawPoints
      \end{tikzpicture}
      \caption{}
      \label{fig:c6:c1}
    \end{subfigure}

    \begin{subfigure}{0.32\textwidth}
      \begin{tikzpicture}
        \cSixPentagon{2,5,3,1,6,4}
        \draw[edge] (1) -- (2) -- (3) -- (4) -- (5) -- (6) -- (1);
        \draw[nonedge] 
          (1) -- (4)
          (1) -- (3);
        \draw[ruleedge] (2) -- (4)
          (2) -- (6)
          (4) -- (6);
        \draw[optionaledge] (3) -- (5)
          (3) -- (6)
          (1) -- (5)
          (2) -- (5);
        \cSixDrawPoints
      \end{tikzpicture}
      \caption{valid if $26\lor\neg24\lor\neg46$}
      \label{fig:c6:d1}
    \end{subfigure}
    \hfill
    \begin{subfigure}{0.32\textwidth}
      \begin{tikzpicture}
        \cSixPentagon{1,5,2,6,4,3}
        \draw[edge] (1) -- (2) -- (3) -- (4) -- (5) -- (6) -- (1);
        \draw[edge] (2) -- (4)
          (2) -- (6);
        \draw[nonedge] (3) -- (6)
          (1) -- (4)
          (1) -- (3);
        \draw[optionaledge] (3) -- (5)
          (1) -- (5)
          (2) -- (5)
          (4) -- (6);
        \cSixDrawPoints
      \end{tikzpicture}
      \caption{}
      \label{fig:c6:e6}
    \end{subfigure}
    \hfill
    \begin{subfigure}{0.32\textwidth}
      \begin{tikzpicture}
        \cSixPentagon{2,6,3,1,5,4}
        \draw[edge] (1) -- (2) -- (3) -- (4) -- (5) -- (6) -- (1);
        \draw[nonedge] (1) -- (4)
          (1) -- (3)
          (2) -- (4);
        \draw[optionaledge](3)--(6)
          (3) -- (5)
          (1) -- (5)
          (2) -- (5)
          (4) -- (6)
          (2) -- (6);
        \cSixDrawPoints
      \end{tikzpicture}
      \caption{}
      \label{fig:c6:e1}
    \end{subfigure}
    \caption{Outside-obstacle representations of 6-cycle case}
    \label{fig:6-cycle-case}
  \end{figure}
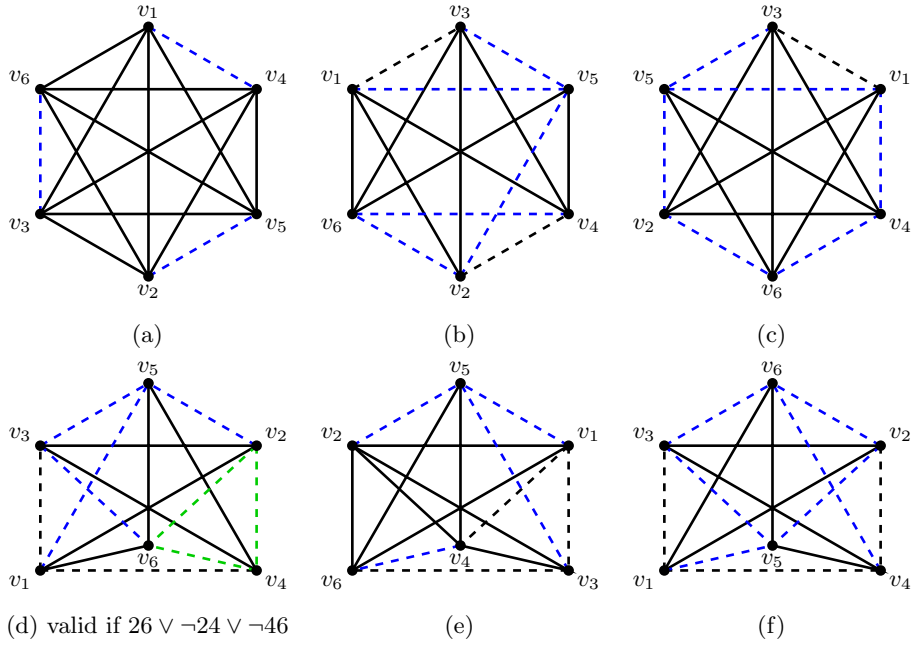
  
Recall that $H$ is a biconnected component of~$G$.
Cases~1, 2, and~3 with $\Circ(H) < 6$ are described in the main body
of the paper.  It remains to deal with case~4, that is, with the case
$\Circ(H)=6$.
In Figs.~\ref{fig:6-cycle-case} to~\ref{fig 6cycle case4 complete},
we use the following convention: black solid/dashed edges mean
determined edges/non-edges;
blue dashed edges can be chosen as edges or non-edges freely.
Green dashed edges can also be chosen as edges or non-edges but it
will be a valid outside-obstacle representation under some
conditions. 
  
    Let $C=v_1v_2v_3v_4v_5v_6 \subset H$ be a 6-cycle.
    If $H$ contains exactly 6 vertices, we can represent it by an
    outside-obstacle representation.
    We use a case distinction to prove this.
    Now to describe the following case distinction, we denote $v_1$ by $1$, and so on.
    Also, $ij$ means $v_i \sim v_j$ and $\neg ij$ means $v_i \not\sim v_j$.
    Note that $ij$ is equivalent to $ji$ and $\neg ij$ to $\neg ji$.
  
    If $13, 24, 35, 46, 51, 62$, we can use the drawing given in
    Fig.~\ref{fig:c6:b1}. 
    Otherwise, \wLOG we assume $\neg13$.
    We distinguish three cases now:
    
    (a) $14$ and $36$.
    If $24$, we use the drawing in Fig.~\ref{fig:c6:c1}.
    If $\neg24$, we use the drawing in Fig.~\ref{fig:c6:a2}.
    
    (b) $\neg14$ and $\neg36$.
    If $24$ and $26$, we use the drawing in Fig.~\ref{fig:c6:d1} or in
    Fig.~\ref{fig:c6:e6} (depending on the vertex which should be on the convex hull of $V(G)$).
    Otherwise, \wLOG $\neg24$ (the case $\neg26$ is symmetric).
    Then we use the drawing in Fig.~\ref{fig:c6:d1} or in
    Fig.~\ref{fig:c6:e1}.
    
    (c) \WLOG we consider only the case $\neg14$ and $36$.
    The other configuration ($14$ and $\neg36$) is symmetric.
    If $24$, we use the drawing in Fig.~\ref{fig:c6:c1}.
    If $\neg24$, we use the drawing in Fig.~\ref{fig:c6:d1} or in
    Fig.~\ref{fig:c6:e1}.\\

    Now suppose that $H$ contains more than six vertices.  We call
    $v_4$ ($v_5$, $v_6$, respectively) an \emph{antipodal} of $v_1$
    ($v_2$, $v_3$, respectively) and vice versa.
  
  We distinguish five subcases of case~(4):
  \begin{enumerate}[(i)]
  \item \emph{There are two vertices in $H \setminus C$ that are
    adjacent to each other.}

    If this is not the case, there are only vertices $x \in H
    \setminus C$ with $N(x) \subset C$.  Then we distinguish the
    following cases.
  \item \emph{There is a vertex in $H \setminus C$ that has only one
      neighbor $C$.}

    For the remaining cases we can assume that every vertex in $H
    \setminus C$ has at least two neighbors.
  \item \emph{There is a vertex in $H \setminus C$ that has at least
      three neighbors in~$C$.}

    Therefore, for the remaining cases, we assume that every $x \in H
    \setminus C$ has exactly two neighbors in~$C$.  These neighbors
    cannot be adjacent on the cycle, as this would imply a longer
    cycle.  So there are only two cases left.
  \item \emph{There is a vertex in $H \setminus C$ whose neighbors are
      antipodals.}
  \item \emph{All vertices in $H \setminus C$ have two non-antipodal
      neighbors on~$C$.}
  \end{enumerate}
  We first show that every vertex $x \in H \setminus C$ is adjacent to
  at least one vertex in $C$.  For contradiction, suppose not.  For
  two vertices $v_i$ and $v_j$ in~$C$, let $C_{ij}$ and $C_{ji}$ be
  parts of~$C$ such that $C=v_iC_{ij}v_jC_{ji}$.  Since $x$ doesn't
  have any neighbors in $C$, biconnectivity of $H$ implies that there
  are vertices $v_i, v_j \in C$ and a $v_i$--$v_j$ path $P$ of length
  at least 4 containing $x$, internally disjoint with $v_iC_{ij}v_j$
  and $v_jC_{ji}v_i$.  Concatenating~$P$ and the longest path among
  $v_iC_{ij}v_j$ and $v_jC_{ji}v_i$, yields a cycle longer than~6.
  
  We make the following observation before starting the case analysis.
    \begin{obs}\label{obs circumference constraint}
  Let $G$ be a graph with $\Circ(G) = k$ and $C$ be a $k$-cycle in $G$.
  If adding an edge $ab$ would create a cycle with a length of more than $k$,
  then we cannot add vertices to $G$ such that an $a-b$ path is formed while
  maintaining the circumference.
  In particular, we cannot add a vertex $x$ to $G$ that is adjacent to both
  $a$ and $b$.
\end{obs}

\noindent\textbf{Case~4(i):} There exist vertices $x, y \in H \setminus C$ such that $x \sim y$.

  \WLOG $x\sim v_1$.
  Let $H'$ be a graph obtained by removing twins from $H$.
  We show that all maximal $v_1$--$v_4$ paths of $H'$ are internally disjoint.

  Observation~\ref{obs cycle and path} implies $x\not\sim v_2, v_6$.
  The same observation shows $y\not\sim v_2, v_3, v_5, v_6$.
  As $H$ is biconnected, there is a path from $y$ to a vertex of $C$ other than $v_1$.
  This path cannot be longer than 1 because otherwise there would be a
  longer cycle.
  This shows $y \sim v_4$.
  Observation~\ref{obs cycle and path} now also implies $x\not\sim v_3, v_5$.

  More generally, we claim that every $v_1$--$v_4$ path has length at most 3.
  For a contradiction, suppose the $v_1$--$v_4$ path $P$ is longer than 3.
  If $P$ is internally disjoint from $v_1v_2v_3v_4$, then $v_1Pv_4v_3v_2v_1$
  forms a cycle longer than 6, so $P$ must contain $v_2$ or $v_3$.
  Similarly, $P$ must contain $v_5$ or $v_6$, hence there is a path between
  $v_2$/$v_3$ and $v_5$/$v_6$ which avoids $v_1$ and $v_4$.
  This fact contradicts Observation~\ref{obs circumference constraint} since
  $v_2,v_3\not\sim v_5,v_6$.

  Let $v_1abv_4$, $v_1cdv_4$, $v_1efv_4$ be internally disjoint $v_1$--$v_4$ paths.
  There are at least three of these paths (using the vertices $v_2$, $v_3$
  or $v_6$, $v_5$ or $x$, $y$).
  An edge ac would create a 7-cycle $v_4fev_1acdv_4$ and an edge ad would
  create an 8-cycle $v_4badcv_1efv_4$,
  so $a~\not\sim c,d$. Same holds for $b$.
  In particular, this implies $v_2 \not\sim v_5, v_6$ and $v_3 \not\sim v_5,
  v_6$. Similarly, we can show that all vertices from different internally
  disjoint $v_1$--$v_4$ paths (including ones of length 2) are pairwise
  non-adjacent.

  Let $z \in H \setminus C$.
  If $z\sim v_1, v_4$, $z$ forms an $v_1$--$v_4$ path of length~2,
  which is too short.
  If $z\not\sim v_1, v_4$ then there exists a $v_1$--$v_4$ path of
  length at least 4, which we showed impossible.
  Otherwise, \wLOG $z\sim v_1$ and $z\not\sim v_4$.
  Let $v_1zwv_4$ be a $v_1$--$v_4$ path containing $z$.
  Since we only need to consider not internally disjoint paths, \wLOG let $v_1z'wv_4$ be another $v_1$--$v_4$ path.
  If $z'\sim v_4$ then it would create a $v_1$--$v_4$ path of length 4, $v_1zwz'v_4$.
  It follows that $z'$ is a twin of $z$.
  Consequently, all $v_1$--$v_4$ paths are internally disjoint after removing twins.

  Since we can handle twins using Lemma~\ref{lem nonadjacent twin}, it's enough
  to provide an outside-obstacle representation with points in convex position
  of a graph whose $v_1$--$v_4$ paths are all internally disjoint.
  Place $v_1$ and $v_4$ arbitrarily.
  Draw a half-circle such that $v_1$ and $v_4$ lie on its diameter; denote its
  center by $O$.
  Assume there are $m$ disjoint paths and put them in an arbitrary order.
  Place vertex $u$, which is on the $i$-th path, on the half-circle so that
  $\angle uv_1O = \frac{2i+1}{2m+1}\pi$ if $u \sim v_1$.
  Otherwise ($u \not\sim v_1, u \sim v_4$), place it so that $\angle uv_1O =
  \frac{2i}{2m+1}\pi$: see Fig.~\ref{fig 6cycle case1} for an example.

  \begin{figure}[tb]
    \begin{subfigure}[b]{0.25\textwidth}
      \centering
      \includegraphics{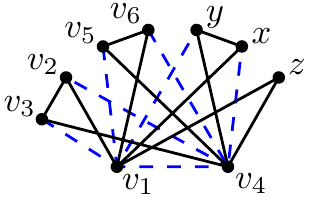}
      \caption{case 4(i) for four paths: $v_1zv_4$, $v_1v_2v_3v_4$,
        $v_1v_6v_5v_4$, $v_1xyv_4$}
      \label{fig 6cycle case1}
    \end{subfigure}
    \hfill
    \begin{subfigure}[b]{0.23\textwidth}
      \centering
      \includegraphics{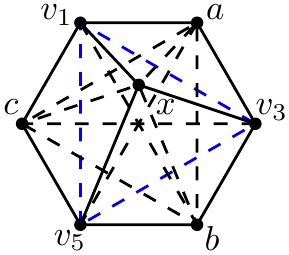}
      \caption{case 4(iii): graph with eliminated twins}
      \label{fig 6cycle case3}
    \end{subfigure}
    \hfill
    \begin{subfigure}[b]{0.24\textwidth}
      \centering
      \includegraphics{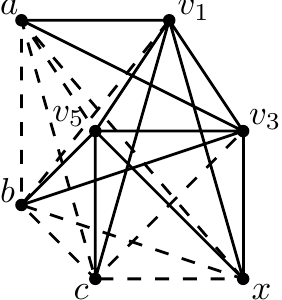}
      \caption{$H[v_1,v_3,v_5]{=}K_3$}
      \label{fig 6cycle case3a}
    \end{subfigure}
    \hfill
    \begin{subfigure}[b]{0.21\textwidth}
      \centering
      \includegraphics{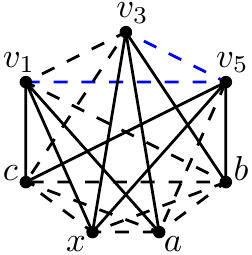}
      \caption{$v_1\not\sim v_3$}
      \label{fig 6cycle case3b}
    \end{subfigure}
    \caption{Cases 4(i) and 4(iii)}
    \label{fig:cases4i+4iii}
  \end{figure}

\medskip
\noindent\textbf{Case~4(ii):} There exists a vertex $x \in H \setminus
C$ that has only one neighbor in~$C$. 

  \WLOG $x\sim v_1$.
  As $H$ is biconnected and using Observation~\ref{obs cycle and path}, there
  exists a $v_1$--$v_4$ path of length~3 containing $x$. Therefore this case
  reduces to case~(i).

\medskip
\noindent\textbf{Case~4(iii):} There exists a vertex $x \in H
\setminus C$ with at least three neighbors in $C$.

  \WLOG $x\sim v_1$.
  By Observation~\ref{obs cycle and path}, $x\sim v_3,v_5$ and $x\not\sim
  v_2,v_4,v_6$. We make the following two observations.
  \begin{enumerate}[(a)]
  \item Since $x$ plays the same role in the 6-cycle
    $v_1xv_3v_4v_5v_6$ as~$v_2$ in~$C$, by the same logic as above,
    $v_2 \not\sim v_4, v_6$.  Similarly, $v_4\not\sim v_6$.
  
  \item Assume that there is another vertex $y \in H \setminus C$.
  If $y \sim v_2$, then $y\not\sim v_4,v_6$
  by Observation~\ref{obs circumference constraint},
  $y\not\sim v_1,v_3$ by Observation~\ref{obs cycle and path},
  and $y\not\sim v_5$ with a path $v_1v_2yv_5$ and 6-cycle
  $v_1xv_3v_4v_5v_6$.
  Since $y$ has at least two neighbors in $C$, this is a contradiction.
  An analogous argument holds for $y \sim v_4$ and $y \sim v_6$.
  Hence, $y$ is adjacent to two or three of the vertices $v_1$, $v_3$, $v_5$.
  \end{enumerate}
  
  Together, (a) and (b) imply that every vertex in $H \setminus \{v_1,
  v_3, v_5\}$ has two or three neighbors among $v_1, v_3, v_5$.
  Note that when $v_1 \sim v_4$, $v_4$ is a twin of~$x$.
  By similar arguments we can exclude $v_2 \sim v_5$ and $v_3 \sim v_6$.
  It follows that the graph~$H'$ (that is, $H$ after removing twins)
  is an induced subgraph of the graph in Fig.~\ref{fig 6cycle case3}.
  In particular, $v_2$ is either a twin of~$a$ or~$x$.  Similar
  statements hold for~$v_4$ and~$v_6$.
  Since the obstacle number of a graph is an upper bound for the obstacle
  number of any induced subgraph,
  it is enough to provide outside-obstacle representations for the graphs shown
  in Fig.~\ref{fig 6cycle case3}.
  If $v_1, v_3, v_5$ are all adjacent, we can use the outside-obstacle
  representation depicted in Fig.~\ref{fig 6cycle case3a} where all vertices
  except from $v_2$ are in convex position.
  Due to symmetry, we can easily change the drawing so that $v_2$ is on the
  convex hull of $V(G)$.
  If at least one pair of $v_1, v_3, v_5$ is non-adjacent, we assume $v_1\not\sim v_3$ without loss of generality and provide the representation in
  Fig.~\ref{fig 6cycle case3b}.

\medskip
\noindent\textbf{Case~4(iv):} There exists a vertex $x \in H \setminus C$ that is adjacent to antipodals.

  \WLOG $x \sim v_1, v_4$.
  An edge $v_2v_5$ would create the 7-cycle $v_2v_5v_6v_1xv_4v_3v_2$ so $v_2 \not\sim v_5$.
  Similarly $v_3\not\sim v_6$.

  Assume there is another vertex $y\in H \setminus C$, which is adjacent to $v_2$.
  If $y \sim v_6$, then $v_2yv_6v_5v_4xv_1v_2$ would be a 7-cycle, so $y \not\sim v_6$.
  Additionally, Observations~\ref{obs circumference constraint} and
  \ref{obs cycle and path} imply $y \not\sim v_1,v_3,v_5$, and it follows that
  $y \sim v_2, v_4$.
  Edges from $v_3$ to $v_1$ or $v_5$ would create a longer cycle.
  So $v_3\not\sim v_5,v_1$ and thus $y$ is a non-adjacent twin of $v_3$.

  Therefore, every $y \in H \setminus C$ is a non-adjacent twin of
  $v_2,v_3,v_5,v_6$ or $x$ and thus by eliminating twins using
  Lemma~\ref{lem nonadjacent twin} we obtain the graph depicted in
  Fig.~\ref{fig 6cycle case4}.
  Depending on the additional edges we have to use one of the three different
  representations depicted in Fig.~\ref{fig 6cycle case4 complete}.
  Note that type a has all vertices in convex position and types b and c
  have two variants with different vertices on the convex hull of $V(G)$.

  Now, if $\neg14$, we use Fig.~\ref{fig 6cycle case4a}.
  If $14,13,24$, we use Fig.~\ref{fig 6cycle case4a}.
  Case $14,15,46$ is symmetric to the previous one.
  So for the remaining cases at least one of $\neg15,\neg46$ and at least one of
  $\neg13,\neg24$ is true.
  If $\neg13,\neg46$ or $\neg15,\neg24$, we use Fig.~\ref{fig 6cycle case4b}.
  \WLOG we can assume $\neg13$ which gives us the configuration $14,\neg13,46,\neg15,24$.
  If $26$ we can use Fig.~\ref{fig 6cycle case4b}, otherwise ($\neg26$) Fig.~\ref{fig 6cycle case4c}.

  \begin{figure}[tb]
    \begin{subfigure}[b]{0.32\textwidth}
      \centering
      \includegraphics{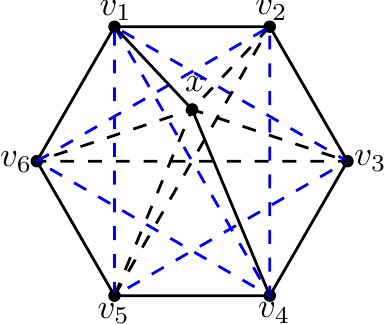}
      \caption{graph with eliminated twins}
      \label{fig 6cycle case4}
    \end{subfigure}
    \hfill
    \begin{subfigure}[b]{0.66\textwidth}
      \centering
      \includegraphics{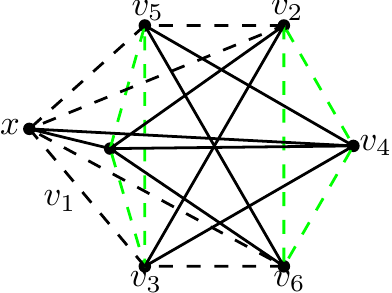}
      \hfill
      \includegraphics{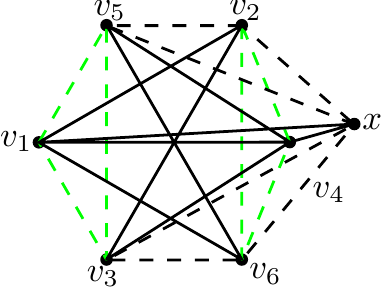}
      \caption{valid if $\neg14\lor15, \neg14\lor46$}
      \label{fig 6cycle case4b}
    \end{subfigure}

    \medskip

    \begin{subfigure}[b]{0.32\textwidth}
      \centering
      \includegraphics{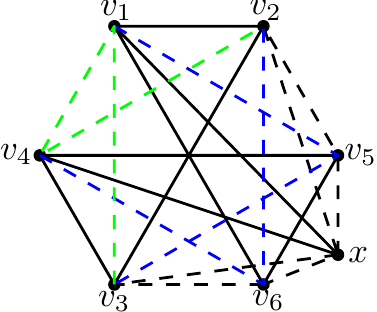}
      \caption{valid if $\neg14\lor13, \neg14\lor24$}
      \label{fig 6cycle case4a}
    \end{subfigure}
    \hfill
    \begin{subfigure}[b]{0.68\textwidth}
      \centering
      \includegraphics{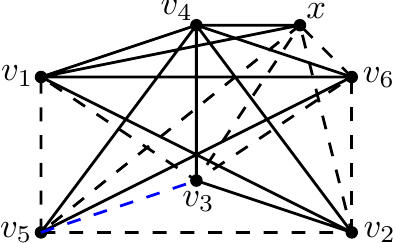}
      \hfill
      \includegraphics{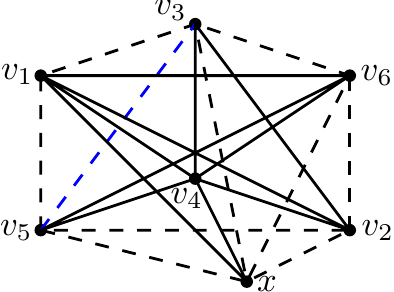}
      \caption{valid if $14,\neg13,\neg15,24,\neg26,46$}
      \label{fig 6cycle case4c}
    \end{subfigure}

    \caption{Case 4(iv)}
    \label{fig 6cycle case4 complete}
  \end{figure}

\noindent\textbf{Case~4(v):} All vertices in $H \setminus C$ have 2 non-antipodal neighbors on $C$.

  Let $x\in H\setminus C$.
  \WLOG $x \sim v_1, v_3$.
  If $v_2$ has three neighbors on the cycle $C' = v_1xv_3v_4v_5v_6$, we apply
  the construction in case (iii) with cycle $C'$ instead of $C$.
  Otherwise $v_2$ has exactly two neighbors ($v_1$ and $v_3$) and thus $x$ is
  its non-adjacent twin.
  Consequently, every additional vertex is a non-adjacent twin of one of the
  vertices $v_1, \dots, v_6$
  and thus removing twins using Lemma~\ref{lem nonadjacent twin} results in a
  graph with 6 vertices, which we already handled.
\end{proof}

\begin{backInTime}{thm graphs order up to 7}
\begin{thm}
\ThmGraphsOrderUpToSevenText
\end{thm}
\end{backInTime}
\begin{proof}
  If the graph does not have $C_7$ as a subgraph, we are done by
  Theorem~\ref{thm graphs circumference up to 6}. 
  Otherwise, we have a 7-vertex graph that contains the 7-cycle
  $v_1v_2v_3v_4v_5v_6v_7$.  We consider 15 types to cover all cases.
  
  Type~1 (Fig.~\ref{fig 7cycle case1}): It is clear that the figure
  is a valid drawing if $\neg13\lor\neg35\lor15,
  \neg35\lor\neg57\lor37, \dots$.
  By moving some vertices depending on the situation, we can make the
  condition tighter. 
  For instance, Fig.~\ref{fig 7cycle case1_2} is a valid drawing even for
  $13,35,\neg15$ if $\neg16\lor36$ and $\neg57\lor37$.
  Hence, we conclude that there is a outside-obstacle representation if there
  do not exist 3 consecutive edges on the convex hull of $V(G)$ in Fig.~\ref{fig 7cycle case1}.
  If there are 3 consecutive edges, \wLOG $v_1v_6$, $v_1v_3$, and $v_3v_5$ are edges.
  This is only a problem if $v_1v_5$ and $v_3v_6$ are non-edges (see Fig.~\ref{fig
  7cycle case1_3}), so we can assume $13, 35, 16, \neg36, \neg15$ for the
  following cases.
  
  Types~2 to 9 cover the case when $47$, and types~10 to 15 cover the
  case when $\neg47$.
  
  Type 2: $47, 24, 26$. When additionally $57$, apply Fig.~\ref{fig 7cycle case2_1}; when $\neg57$, apply Fig.~\ref{fig 7cycle case2_2}.
  
  Type 3: $47, 25, 27$. Symmetric to type 2.
  
  Type 4: $47, \neg24, \neg25$. Fig.~\ref{fig 7cycle case4}.
  
  Type 5: $47, \neg26, \neg27$. Symmetric to type 4.
  
  Type 6: $47,24,\neg27$. Fig.~\ref{fig 7cycle case6}.
  
  Type 7: $47,\neg24,27$. Symmetric to type 6.
  
  Type 8: $47,\neg24,25,26,\neg27$. When additionally $37$, apply Fig.~\ref{fig 7cycle case8_1}; when $\neg37$, apply Fig.~\ref{fig 7cycle case8_2}.
  
  \begin{figure}[tb]
    \begin{subfigure}[b]{0.3\textwidth}
      \centering
      \includegraphics{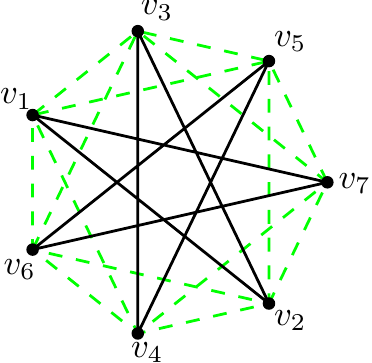}
      \caption{type 1}
      \label{fig 7cycle case1}
    \end{subfigure}
    \hfill
    \begin{subfigure}[b]{0.3\textwidth}
      \centering
      \includegraphics{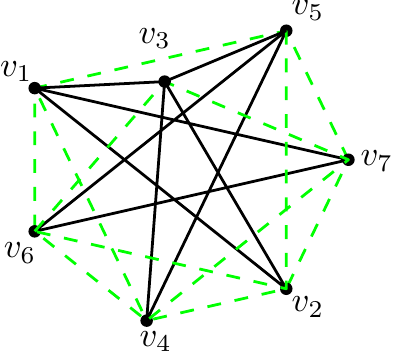}
      \caption{still valid}
      \label{fig 7cycle case1_2}
    \end{subfigure}
    \hfill
    \begin{subfigure}[b]{0.37\textwidth}
      \centering
      \includegraphics{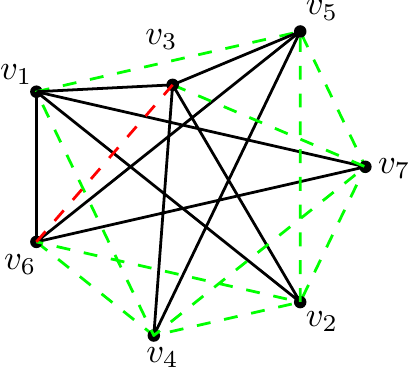}
      \caption{invalid due to the red non-edge}
      \label{fig 7cycle case1_3}
    \end{subfigure}
  
    \medskip

    \begin{subfigure}[b]{0.32\textwidth}
      \centering
      \includegraphics{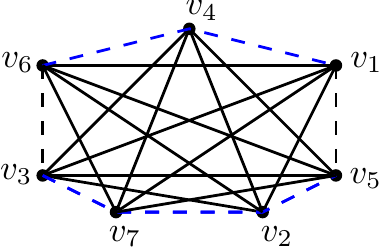}
      \caption{type 2.1}
      \label{fig 7cycle case2_1}
    \end{subfigure}
    \hfill
    \begin{subfigure}[b]{0.32\textwidth}
      \centering
      \includegraphics{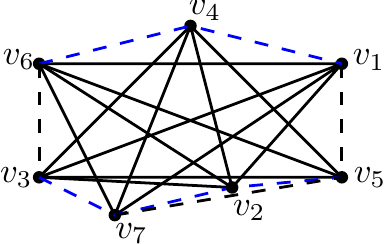}
      \caption{type 2.2}
      \label{fig 7cycle case2_2}
    \end{subfigure}
    \hfill
    \begin{subfigure}[b]{0.32\textwidth}
      \centering
      \includegraphics{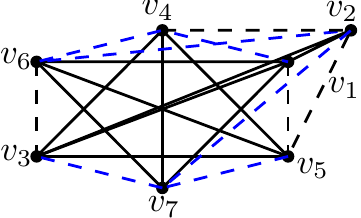}
      \caption{type 4}
      \label{fig 7cycle case4}
    \end{subfigure}
  
    \medskip

    \begin{subfigure}[b]{0.32\textwidth}
      \centering
      \includegraphics{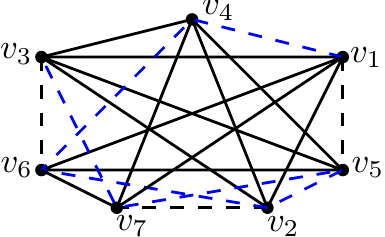}
      \caption{type 6}
      \label{fig 7cycle case6}
    \end{subfigure}
    \hfill
    \begin{subfigure}[b]{0.32\textwidth}
      \centering
      \includegraphics{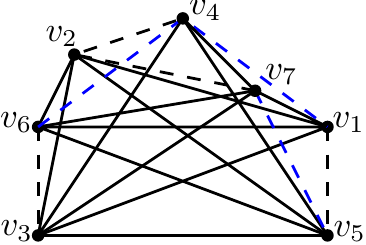}
      \caption{type 8.1}
      \label{fig 7cycle case8_1}
    \end{subfigure}
    \hfill
    \begin{subfigure}[b]{0.32\textwidth}
      \centering
      \includegraphics{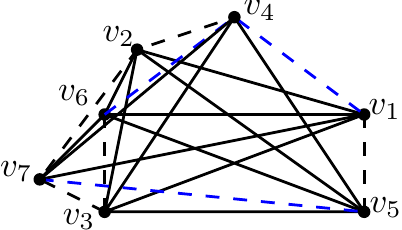}
      \caption{type 8.2}
      \label{fig 7cycle case8_2}
    \end{subfigure}
    \caption{Outside-obstacle representations of 7-cycle case: types~1 to~8}
  \end{figure}

  Type 9: $47,24,\neg25,\neg26,27$. Fig.~\ref{fig 7cycle case9}.
  
  Type 10: $\neg47,24,27$. Fig.~\ref{fig 7cycle case10}.
  
  Type 11: $\neg47,\neg24,\neg27$. Fig.~\ref{fig 7cycle case11}.
  
  Due to symmetry, we can assume $\neg47,\neg24,27$ for the rest.
  
  Type 12: $\neg47,\neg24,27,\neg25$. Fig.~\ref{fig 7cycle case12}.
  
  Type 13: $\neg47,\neg24,27,\neg57$. Fig.~\ref{fig 7cycle case13}.
  
  Type 14: $\neg47,\neg24,27,25,57,26$. Fig.~\ref{fig 7cycle case14}.
  
  Type 15: $\neg47,\neg24,27,25,57,\neg26$. Fig.~\ref{fig 7cycle case15}.
    
  \begin{figure}[tb]
    \begin{subfigure}[b]{0.325\textwidth}
      \centering
      \includegraphics{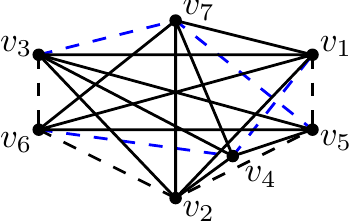}
      \caption{type 9}
      \label{fig 7cycle case9}
    \end{subfigure}
    \hfill
    \begin{subfigure}[b]{0.325\textwidth}
      \centering
      \includegraphics{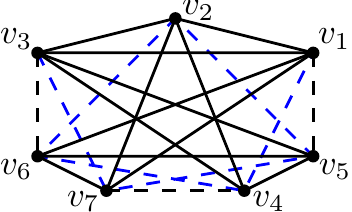}
      \caption{type 10}
      \label{fig 7cycle case10}
    \end{subfigure}
    \hfill
    \begin{subfigure}[b]{0.325\textwidth}
      \centering
      \includegraphics{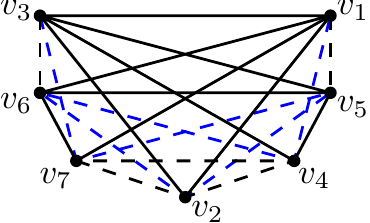}
      \caption{type 11}
      \label{fig 7cycle case11}
    \end{subfigure}
  
    \medskip
    
    \begin{subfigure}[b]{0.24\textwidth}
      \centering
      \includegraphics{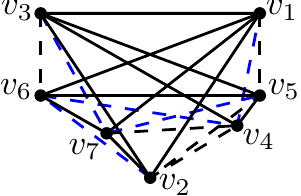}
      \caption{type 12}
      \label{fig 7cycle case12}
    \end{subfigure}
    \hfill
    \begin{subfigure}[b]{0.24\textwidth}
      \centering
      \includegraphics{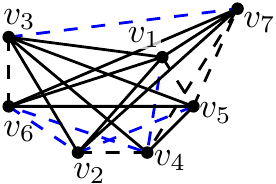}
      \caption{type 13}
      \label{fig 7cycle case13}
    \end{subfigure}
    \hfill
    \begin{subfigure}[b]{0.24\textwidth}
      \centering
      \includegraphics{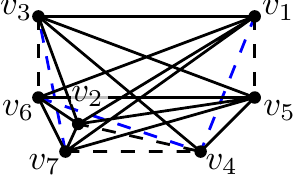}
      \caption{type 14}
      \label{fig 7cycle case14}
    \end{subfigure}
    \hfill
    \begin{subfigure}[b]{0.24\textwidth}
      \includegraphics{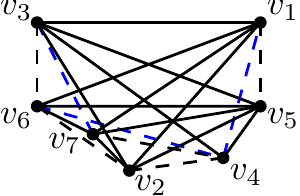}
      \caption{type 15}
      \label{fig 7cycle case15}
    \end{subfigure}

    \caption{Outside-obstacle representations of 7-cycle case: types~9 to~15}
  \end{figure}
\end{proof}

\section{Missing Proofs of Section~\ref{sec:co-bipartite}}
\label{appendixB}

This appendix contains the full details showing that $B_8$ has no
inside-obstacle representation, as formalized in
Lemma~\ref{lem:8-vertex-no-inside-obstacle}.  To this end, we first
establish some useful properties of graphs with inside-obstacle
representations.

\begin{obs}\label{obs convex hull cycle}
  In an inside-obstacle representation of a graph $G=(V,E)$, the
  vertices on $\CH(V)$ form a cycle.
\end{obs}

\begin{obs}\label{obs 4-path}
In an inside-obstacle representation of a graph $G$, if $G$ contains a 
3-edge induced path $uxyv$ where $u$ and $v$ are on the convex hull of $V(G)$ and 
$x$ and $y$ are not on the convex hull of $V(G)$, then the line segments 
$\overline{ux}$ and $\overline{yv}$ do not intersect and the 
quadrilateral $uxyv$ is convex.
\end{obs}
\begin{proof}
Suppose that the line segments $\overline{ux}$ and $\overline{yv}$ 
intersect, and let $z$ be the intersection point. Let $P$ be a chain 
on the convex hull of $V(G)$ such that the region bounded by $uzvPu$ contains 
the line segment $\overline{uv}$ of a non-edge $uv$. The obstacle should 
lie inside the region $uzvPu$ due to $\overline{uv}$. However, the line 
segment $\overline{uy}$ of the non-edge $uy$ lies completely outside the 
region $uzvPu$. Thus contradicting $\Obsin(G)=1$.

Therefore $\overline{ux}$ and $\overline{yv}$ do not intersect, i.e., 
$uxyv$ forms a non-intersecting quadrilateral. Let $P$ be a chain on the
convex hull of $V(G)$ such that the region bounded by $uxyvPu$ contains the line 
segment $\overline{uy}$. Notice that, when $uxyv$ is not convex, the line 
segment $\overline{yx}$ lies outside this region, i.e., contradicting 
$\Obsin(G)=1$. Thus, the quadrilateral $uxyv$ is convex.
\end{proof}

\begin{obs}\label{obs two 4-paths}
Let $G$ be a graph which contains the vertices $u,v,u',v',x,y$ such that 
$uxyv$ and $u'xyv'$ are induced 3-edge paths.
In an inside-obstacle representation of a graph $G$, if $u, v, u', v'$ are
on the convex hull of $V(G)$, $x$ and $y$ are not on the convex hull of $V(G)$,
and $\overline{uv}$ and $\overline{u'v'}$ intersect, then $u,v'$ 
or $v,u'$ are not consecutive on the convex hull of $V(G)$. If additionally $u,u',v,v'$ 
are consecutive, then neither $x$ nor $y$ are contained in the quadrilateral 
formed by $u,u',v,v'$.
\end{obs}
\begin{proof}
Consider the ray $\overrightarrow{ux}$, and let $p$ be the intersection 
point between this ray and the convex hull of $V(G)$. Further, let $P$ be the chain 
on the convex hull that connects $p$ to $v$ but does not contain $u$. 
Similarly, let $q$ be the intersection point of the ray 
$\overrightarrow{u'x}$ with the convex hull and let $Q$ be the chain on 
the convex hull that connects $q$ to $v'$ but does not contain $u'$. 
By Observation~\ref{obs 4-path}, $uxyv$ is convex, i.e., $y$ is inside 
the region bounded by $xvPpx$. Similarly, $u'xyv'$ is convex, i.e., $y$ is
inside the region bounded by $xv'Qqx$. Thus, the regions $xvPpx$ and 
$xv'Qqx$ intersect and it follows that $P$ and $Q$ overlap.

Suppose both of $u, v'$ and $v, u'$ are consecutive. Since $u\not\sim v$ 
,$u'\not\sim v'$, 
and $\overline{uv}$ and $\overline{u'v'}$ intersect, their order on the convex hull is
$uv'\dots{}vu'$. If $x$ is contained in the quadrilateral $uv'vu'$, then
the points $p$ and $q$ occur on the convex hull such that: both $p$ and 
$q$ are between $u$ and $u'$ (i.e., $u\dots{}p\dots{}q\dots{}vu'$), 
$p$ is between $v'$ and $u'$ (i.e., $uv'\dots{}p\dots{}u'$), and $q$ is 
between $u$ and $v$ (i.e., $u\dots{}q\dots{}vu'$). However, these 
conditions contradict the fact that the chains $P$ and $Q$ overlap. 
When $x$ is not contained in the quadrilateral $uv'vu'$, we have the 
ordering is $uv'\dots{}vu'\dots{}p\dots{}q\dots{}u$, and, again, $P$ 
and $Q$ do not overlap. Thus, one of $u,v'$ and $v,u'$ are non-consecutive.

Now suppose that $u, u', v, v'$ are consecutive. Note that, \wLOG they are 
consecutive in that order by the previous paragraph and since 
$u\not\sim v, u'\not\sim v'$. If $x$ is contained in the quadrilateral 
$uu'vv'$, we have the convex hull is ordered so that $uu'\dots{}p\dots{}q$
and $uu'\dots{}p\dots{}v'$ on the convex hull, i.e., causing $P$ and $Q$ to
not overlap. Symmetrically, $y$ is also not contained in the quadrilateral.
\end{proof}

Using the above observations we proceed with the main lemma. 

\begin{backInTime}{lem:8-vertex-no-inside-obstacle}
\begin{lem}
  \LemEightVertexNoInsideObstacle
\end{lem}
\end{backInTime}
\begin{proof}
  Suppose $\Obsin(B_8) = 1$ for contradiction. The following observation, 
  together with Observation~\ref{obs convex hull cycle}, greatly restricts 
  the vertices which can occur on $\CH(V(B_8))$.

\begin{obs}\label{obs vtx on CH}
Assume that $B_8$ has an inside-obstacle representation.
Then the order of the vertices on the convex hull of $V(G)$ satisfies the following restrictions.
If the two vertices $\{v_2, v_4\}$ lie on the convex hull of $V(G)$,
then they must be consecutive.
The same holds for each of the pairs $\{v_1, v_4\}, \{v_3, v_4\}, 
\{v_5, v_8\}, \{v_6, v_8\}, \{v_7, v_8\}$, for any maximal subset of 
$\{v_1, v_2, v_3\}$ occurring on the convex hull of $V(G)$, and for any maximal 
subset of $\{v_5, v_6, v_7\}$ occurring on the convex hull of $V(G)$.
\end{obs}
\begin{proof}
Suppose $v_2$ and $v_4$ are both on the convex hull of $V(G)$. If they are not 
consecutive, then there exist vertices $x,y$, each distinct from $v_2$ and
$v_4$ such that $x$ and $y$ occur in distinct chains ($P$ and $Q$ resp.) 
connecting $v_2$ and $v_4$ along the convex hull. 
Suppose $v_5$ is also on the convex hull, and that $x=v_5$. Since $v_5$ 
is adjacent to all but $v_2$ and $v_4$, $x$ is adjacent to $y$, but now we
see that the line segments $\overline{v_5v_2}$ and $\overline{v_5v_4}$ are 
separated into disjoint bounded regions by the chord $\overline{xy}$ and the
convex hull. Similarly, when $v_5$ is not on the convex hull, we have the 
segments $\overline{v_5x}$ and $\overline{v_5y}$ which again, together with 
the convex hull separate $\overline{v_5v_2}$ and $\overline{v_5v_4}$ into 
disjoint bounded regions. 
Hence, in either case, those two non-edges cannot be obstructed by one 
inside obstacle.

We can similarly use $v_6$ for $\{v_1, v_4\}$, $v_7$ for $\{v_3, v_4\}$, 
$v_2$ for $\{v_5, v_8\}$, $v_1$ for $\{v_6, v_8\}$, $v_3$ for 
$\{v_7, v_8\}$, $v_8$ for any subset of $\{v_1, v_2, v_3\}$, and $v_4$ 
for any subset of $\{v_5, v_6, v_7\}$.
\end{proof}
   
By Observation~\ref{obs vtx on CH}, only two of $v_1, v_2, v_3$ can be 
consecutive with $v_4$, i.e., at least one of $v_1, v_2, v_3, v_4$ is not 
on the convex hull. Symmetrically, at least one of $v_5, v_6, v_7, v_8$ 
is not on the convex hull. We now consider the different cases regarding
size of the convex hull. 

\medskip  
   
\noindent\textbf{Case~1:} $\CH(V(B_8))$ is a 6-gon.\newline
  Due to symmetry, we only need to consider 3 cases: when $v_1, v_5$ are 
  interior to the convex hull, when $v_1, v_8$ are interior, and when 
  $v_4, v_8$ are interior. For the first and second case, by 
  Observation~\ref{obs vtx on CH}, $v_2, v_3, v_4$ are on the convex hull.
  Thus, $v_2, v_3$ are consecutive, $v_3, v_4$ are consecutive, and $v_2, v_4$
  are consecutive, which contradicts $\CH(V(B_8))$ being a 6-gon.
  For the third case (when $v_4, v_8$ are interior), we have $v_1, v_2, v_3$ 
  consecutive and $v_5, v_6, v_7$ consecutive. \WLOG suppose $v_1, v_5$ are 
  consecutive. Three orderings are possible on the convex hull: 
  $v_1v_2v_3v_6v_7v_5$, $v_1v_3v_2v_7v_6v_5$, and $v_1v_3v_2v_6v_7v_5$. 
  For $v_1v_2v_3v_6v_7v_5$, the non-edge $v_3v_7$ lies inside the quadrilateral 
  $v_2v_3v_6v_7$ while another non-edge $v_2v_5$ lies inside the quadrilateral 
  $v_2v_1v_5v_7$. However, these quadrilaterals do not intersect, yielding a 
  contradiction. Note that, $v_1v_3v_2v_7v_6v_5$ also has two non-edges 
  $v_3v_7$ and $v_1v_6$, occurring within the non-overlapping quadrilaterals 
  $v_3v_2v_7v_6$ and $v_3v_1v_5v_6$ (respectively). We now apply
  Observation~\ref{obs two 4-paths} on the two induced paths $v_1v_4v_8v_6$ and
  $v_2v_4v_8v_5$. This shows that $v_1v_5$ or $v_2v_6$ should be not consecutive, 
  invalidating the ordering $v_1v_3v_2v_6v_7v_5$.
  
\medskip
  
\noindent\textbf{Case~2:} $\CH(V(B_8))$ is a 5-gon.\\
  \WLOG we choose three vertices from $v_1, v_2, v_3, v_4$ and two
  vertices from $v_5, v_6,
  v_7, v_8$ to be on the convex hull. When $v_1$ is not on the convex hull, 
  $v_2, v_3, v_4$ are on the convex hull, which was the case already 
  rejected above. Symmetrically, it is not possible to omit $v_2$ or $v_3$ 
  from the convex hull. 
  When $v_4$ is not on the convex hull, $v_1, v_2, v_3$ are consecutive by 
  Observation~\ref{obs vtx on CH}. \WLOG the ordering on the convex hull 
  is $v_1v_2v_3$. There are 3 possible orderings: $v_1v_2v_3v_5v_7$, 
  $v_1v_2v_3v_6v_5$, $v_1v_2v_3v_6v_7$. 
  First, note that $v_1v_2v_3v_5v_7$ has two induced paths $v_3v_4v_8v_7$ 
  and $v_2v_4v_8v_5$. Thus, by Observation~\ref{obs two 4-paths}, $v_4$ and 
  $v_8$ are inside the triangle $v_1v_2v_7$, i.e., a non-edge $v_1v_8$ lies 
  inside $v_1v_2v_7$. However, the region $v_2v_3v_5v_7$ is non-overlapping 
  with $v_1v_2v_7$, but contains the non-edge $v_2v_5$. 
  For the next case (i.e., $v_1v_2v_3v_6v_5$), we note the two induced paths
  $v_1v_4v_8v_6$ and $v_2v_4v_8v_5$. Thus, $v_4$ and $v_8$ are inside the 
  triangle $v_2v_3v_6$. In particular, we have the non-edge $v_8v_2$ inside 
  $v_2v_3v_6$ and the non-edge $v_2v_5$ inside $v_2v_6v_5v_1$. 
  Finally, for $v_1v_2v_3v_6v_5$, the two induced paths $v_1v_4v_8v_6$ and 
  $v_3v_4v_8v_7$ fail the condition in Observation~\ref{obs two 4-paths}.
  
\medskip
  
\noindent\textbf{Case~3:} $\CH(V(B_8))$ is a 4-gon.\\
  First, we consider the case when three vertices from $v_1, v_2, v_3, v_4$ 
  and one vertex from $v_5, v_6, v_7, v_8$ is on the convex hull. As before, 
  the three vertices are $v_1, v_2, v_3$. \WLOG we assume $v_1v_2v_3$ is the 
  order on the convex hull, i.e., the forth vertex must be $v_5$. However, 
  now either $v_1v_2v_6v_5$ or $v_3v_2v_6v_5$ forms a \emph{dart}, so it 
  cannot be represented by 1 obstacle (see Fig.~4 in~\cite{akl-ong-DCG10}).
  
  We now have 2 vertices from $v_1, v_2, v_3, v_4$ and 2 from $v_5,
  v_6, v_7, v_8$ on the convex hull. 
  Suppose $v_4$ is interior and \wLOG $v_1$ and $v_2$ are on the convex hull.
  There are 3 cases to consider: $v_1v_2v_6v_5$, $v_1v_2v_6v_7$,
  $v_1v_2v_7v_5$. Since there are only 4 vertices on the convex hull,
  it suffices to find two induced paths for Observation~\ref{obs two 4-paths} 
  to provide contradiction. 
  For the first case, we find $v_1v_4v_8v_6$ and $v_2v_4v_8v_5$.
  For the second case, we observe an induced path $v_1v_4v_8v_6$ and by 
  Observation~\ref{obs 4-path} $v_4$ and $v_8$ lie on the same side of 
  $v_1v_6$. Hence, the non-edges $v_4v_7$ and $v_2v_8$ lie inside the 
  non-overlapping regions $v_1v_4v_8v_6v_2$ and $v_1v_4v_8v_6v_7$ 
  (respectively). 
  For the third case, we apply the same logic but starting from the induced
  path $v_2v_4v_8v_5$. 
  
  Finally, we have $v_4$ on the convex hull and suppose \wLOG that $v_1$ is 
  as well. This means the convex hull is either $v_1v_4v_8v_5$ or 
  $v_1v_4v_8v_7$. By symmetry, it suffices to consider $v_1v_4v_8v_5$. 
  Here we again observe the two forbidden induced paths: $v_1v_2v_6v_8$ and 
  $v_4v_2v_6v_5$. 

\medskip

\noindent\textbf{Case~4:} $\CH(V(B_8))$ is a 3-gon.\\
  \WLOG the convex hull either consists of 3 vertices from $v_1, v_2, v_3, v_4$ 
  or 2 vertices from $v_1, v_2, v_3, v_4$ and 1 vertex from $v_5, v_6, v_7, v_8$.
  Up to symmetry, it suffices to consider 3 cases: $v_1v_2v_3$, $v_1v_2v_4$, $v_1v_2v_7$.
  
  When the convex hull is $v_1v_2v_3$, we note the induced 4-cycle $v_1v_2v_6v_5$. 
  Suppose $\overline{v_1v_6}$ and $\overline{v_2v_5}$ intersect, i.e. $v_1v_2v_6v_5$ is a convex 
  quadrilateral. Since the region $v_1v_2v_6v_5$ contains a non-edge $v_1v_6$ and 
  $v_1v_8$, $v_5v_4$ are non-edges, $v_4$ and $v_8$ lie inside $v_1v_2v_6v_5$. 
  However, by Observation~\ref{obs two 4-paths}, the two induced paths $v_1v_4v_8v_6$ 
  and $v_2v_4v_8v_5$ are forbidden.
  Consequently $\overline{v_1v_6}$ and $\overline{v_2v_5}$ do 
  not intersect. However, in this case, $v_1v_6$ and $v_2v_5$ occur in disjoint
  bounded regions as depicted in Fig.~\ref{fig 3gon_1}.
  
  When the convex hull is $v_1v_2v_4$, we also observe an induced 4-cycle: 
  $v_2v_7v_8v_4$. Suppose it forms a convex quadrilateral. Since $v_3v_7$ and
  $v_5v_4$ are non-edges, they lie inside $v_2v_7v_8v_4$. However, again,
  by Observation~\ref{obs two 4-paths}, we have forbidden induced paths $v_2v_3v_5v_8$ 
  and $v_4v_3v_5v_7$.
  If $\overline{v_2v_8}$ and $\overline{v_7v_4}$ do not intersect, then $v_4v_7$ and $v_2v_8$
  occur in disjoint bounded regions as depicted in Fig.~\ref{fig 3gon_2}.
  
  \begin{figure}[tb]
    \begin{subfigure}{0.32\textwidth}
      \includegraphics{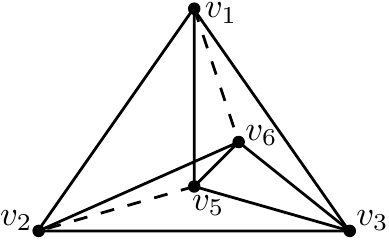}
      \caption{}
      \label{fig 3gon_1}
    \end{subfigure}
    \begin{subfigure}{0.32\textwidth}
      \includegraphics{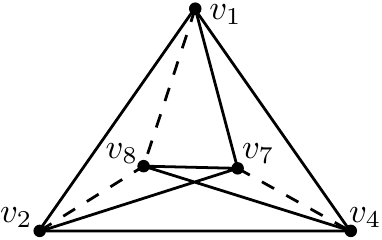}
      \caption{}
      \label{fig 3gon_2}
    \end{subfigure}
    \begin{subfigure}{0.32\textwidth}
      \includegraphics{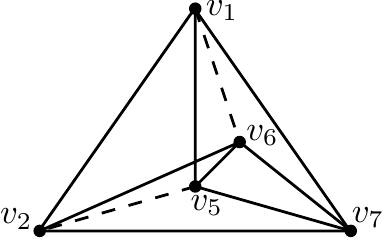}
      \caption{}
      \label{fig 3gon_3}
    \end{subfigure}
    \caption{When $\CH(V(B_8))$ is a 3-gon}
  \end{figure}
  
  When the convex hull is $v_1, v_2, v_7$, we conclude that $\overline{v_1v_6}$ and
  $\overline{v_2v_5}$ do not intersect similarly to the first case. However, again, 
  similarly to the first case, we see that when $\overline{v_1v_6}$ and $\overline{v_2v_5}$ do not
  intersect, $v_1v_6$ and $v_2v_5$ occur in disjoint bounded regions as depicted in 
  Fig.~\ref{fig 3gon_3}.
\end{proof}

\section{Missing Proofs of Section~\ref{sec:hardness}}
\label{appendixC}

In this appendix we describe how to use the NP-hardness of the outside-obstacle 
sandwich problem to show NP-hardness for both the single-obstacle 
sandwich problem and the inside-obstacle sandwich problem. 
To this end, we first make an observation. 

\begin{obs}
  \label{obs:combining_graphs}
  Let $G$ be any graph, and let \Gin be a graph with $\Obsin(\Gin)=1$,
  but $\Obsout(\Gin)>1$.  If $G^*$ is the disjoint union of \Gin and
  $G$ (i.e., $V(G^*) = V(\Gin) \cup V(G)$ and $E(G) = E(\Gin) \cup
  E(G)$), then the following properties hold:
  \begin{enumerate}
  \item \label{prop:combining_graphs1} $\Obsout(G^*) >1$. 
  \item \label{prop:combining_graphs2} In every inside-obstacle
    representation of $G^*$, the point set of $G$ is contained inside
    the convex hull of the obstacle, i.e., any single-obstacle
    representation of $G^*$ contains an outside-obstacle
    representation of $G$.
  \end{enumerate}
\end{obs}
\begin{proof}
  Property~\ref{prop:combining_graphs1} is clear since any
  outside-obstacle representation of $G^*$ would certainly contain an
  outside-obstacle representation of \Gin, contradicting
  $\Obsout(\Gin)>1$.

  For Property~\ref{prop:combining_graphs2}, suppose that $G$ has an
  inside-obstacle representation and let $P$ be the obstacle. Notice
  that $P$ must be strictly contained within the convex hull of the
  point set of~\Gin (otherwise we would have an outside-obstacle
  representation of~\Gin). In particular, $P$ is contained in a region
  whose boundary consists of (parts of) edges of~\Gin.  Now consider any
  vertex~$v$ of~$G$, and suppose for a contradiction that $v$ is not
  placed within the convex hull of~$P$. This means that there is a
  line~$\ell$ such that
\begin{itemize}
\item $\ell$ intersects the boundary of $P$,
\item the interior of $P$ is contained in one open half-plane $h^+$ defined 
 by $\ell$, and 
\item $v$ is contained in the other open half-plane $h^-$.
\end{itemize}
However, as $P$ is an inside obstacle of \Gin. there must be a vertex $u$ 
of \Gin that is contained in $h^-$, contradicting the fact that 
the obstacle $P$ must intersect the line $\overline{uv}$.
\end{proof}

From this observation, we can extend the NP-hardness of the outside-obstacle
sandwich problem to both the inside-obstacle sandwich problem and the 
single-obstacle sandwich problem.

\begin{cor}\label{cor:inside_and_single_NP-hard}
The inside-obstacle graph sandwich problem and the single-obstacle
graph sandwich problem are both NP-hard. These problems remain NP-hard even
when restricted to sandwich instances $(G,H)$ where $G$ is connected. 
\end{cor}
\begin{proof}
  Let $(G,H)$ be an instance of the outside-obstacle graph sandwich
  problem.  Recall that, for the graph $B_{11}$ (given in
  Fig.~\ref{fig:K}), we have $\Obsin(B_{11})=1$ and
  $\Obsout(B_{11})>1$.  From the pair $(G \cup B_{11}$, $H \cup
  B_{11})$, we make $3|V(G)|$ instances $\mathcal{I}_1, \ldots,
  \mathcal{I}_{3n}$ of the single-obstacle sandwich problem where each
  instance $\mathcal{I}_i = (G \cup B_{11} \cup \{u_iv_i\}$, $H \cup
  B_{11} \cup \{u_iv_i\})$ is formed by adding a single edge~$u_iv_i$
  connecting a vertex~$u_i$ of~$G$ to a vertex~$v_i$ of~$B_{11}$.  Due
  to the symmetry in~$B_{11}$, there are at most $3n$ non-isomorphic
  ways to add such an edge.

  We now claim that $(G,H)$ has a solution to the outside-obstacle
  sandwich problem if and only if some single-obstacle sandwich
  instance $\mathcal{I}_i$ has a solution.  Moreover, a solution to
  $\mathcal{I}_i$ is always an inside-obstacle representation.  This
  proves the statement of the corollary.  It remains to show that our
  claim holds.

  For the forward direction, let $G'$ be a solution to the
  outside-obstacle sandwich instance $(G,H)$. We can place $G'$
  ``inside'' the obstacle of the inside-obstacle representation of
  $B_{11}$ (e.g., the one depicted in Fig.~\ref{fig:K}) to obtain an
  inside-obstacle representation of $B_{11} \cup G'$. Furthermore, we
  can make a thin ``tunnel'' into the obstacle so that we realize
  precisely one edge connecting $G'$ and $B_{11}$.  This provides an
  inside-obstacle representation of one of the instances
  $\mathcal{I}_i$.

  For the reverse direction, consider an instance $\mathcal{I}_i = (G
  \cup B_{11} \cup \{u_iv_i\}$, $H \cup B_{11} \cup \{u_iv_i\})$ that
  has a solution~$G'$.  Recall that $u_i \in V(G)$ and $v_i \in
  V(B_{11})$.  Let $P$ be the corresponding obstacle.  Note that, by
  Observation~\ref{obs:combining_graphs}, $G'$ also provides a
  single-obstacle representation of $G' \setminus \{u_i\}$ using $P$,
  and this must be an inside-obstacle representation.  Moreover,
  $G'[V(G) \setminus \{u_i\}]$ is contained in the convex hull
  of~$P$.  In particular, $P$ is an outside obstacle of $G'[V(G)]$,
  and $P$ is an inside obstacle of $G'$.
\end{proof}

\section{Inside-Obstacle Number~2, but Outside-Obstacle Number~1}
\label{appendixD}

Since inside-obstacles cannot pierce the convex hull of a drawing, 
every graph with an inside-obstacle representation must either be complete
or contain a cycle. 
This appendix introduces a non-trivial graph (i.e., containing a cycle) that 
has outside-obstacle number 1 but inside-obstacle number greater than 1.
It is trivial that every 4-vertex graph $G$ which contains a cycle 
has $\Obsin(G)=1$.
On the other hand, we will show that there is a unique 5-vertex graph with 
$\Obsin(G)>1$ and, as such, the 4-vertex observation is tight.

\begin{thm}\label{thm:outside1-inside>1}
  Among all 5-vertex graphs, $K_{2,3}$ is the unique graph that
  contains a cycle, has a outside-obstacle number~1, and
  inside-obstacle number greater than~1.
\end{thm}
\begin{proof}
  Let $(\{u_1, u_2\},\{v_1,v_2,v_3\})$ be a bipartition of $K_{2,3}$.
  By Theorem~\ref{thm graphs order up to 7}, $K_{2,3}$ has an outside-obstacle representation.
  To prove that $\Obsin(K_{2,3}) > 1$, we assume $\Obsin(K_{2,3})=1$ for contradiction.
  Since the convex hull should form a cycle in an inside-obstacle representation by Observation~\ref{obs convex hull cycle}, \wLOG we assume the cycle $u_1v_1u_2v_2$ is the convex hull.
  By placing $v_3$ inside quadrilateral $u_1v_1u_2v_2$, we notice that non-edges $v_1v_3$ and $v_2v_3$ lie inside the different bounded regions so a single inside obstacle cannot block both.
  
  To prove uniqueness, let $G$ be a 5-vertex graph not isomorphic to $K_{2,3}$. It is enough to provide an inside-obstacle representation for the connected graphs with no leaves. Since $G$ doesn't have a leaf and isn't isomorphic to $K_{2,3}$, $G$ contains a 5-cycle $C$. To make an inside-obstacle representation of $G$, we place the points of $C$ as a regular 5-gon. Notice that the diagonals of $C$ make a star-shape inside the 5-gon which, in turn, provides an inner 5-gon $P$ where each side corresponds to a diagonal of $C$. We can use $P$ as an obstacle by simply bending each side outward when the corresponding diagonal of $C$ is a non-edge. 
  
  This finishes the proof of Theorem~\ref{thm:outside1-inside>1}.
\end{proof}

\end{document}